\documentclass[10pt,oneside,a4paper]{article}
\usepackage{mathrsfs}
\usepackage{mathdots}
\usepackage{amsmath}
\usepackage{amsthm}
\usepackage{bbm}
\usepackage{subfloat}
\usepackage{subfigure}
\usepackage{fancyhdr,graphicx}
\usepackage{amsfonts}
\usepackage{amssymb}
\usepackage{latexsym,bm}
\usepackage{newlfont}
\usepackage{latexsym}
\usepackage{caption2}
\usepackage{algorithmic,algorithm}
\usepackage{multirow}
\usepackage[pagewise]{lineno}
\usepackage{slashbox}
\usepackage{booktabs}
\usepackage{float}
\usepackage{array}
\usepackage{tabularx}
\newtheorem{definition}{Definition}
\newtheorem{theorem}{Theorem}
\newtheorem{lemma}[theorem]{\quad Lemma}
\newtheorem{example}{Example}
\newtheorem{remark}{Remark}

\numberwithin{figure}{section} \numberwithin{equation}{section}

\makeatletter
\long\def\@makecaption#1#2{%
 \vskip\abovecaptionskip
  \sbox\@tempboxa{{#1.}\quad #2}%
 \ifdim \wd\@tempboxa >\hsize
    { #1.}\quad #2\par
     \else
  \global \@minipagefalse
   \hb@xt@\hsize{\hfil\box\@tempboxa\hfil}%
   \fi
   \vskip\belowcaptionskip}
\makeatother

\title{ \textbf{High order schemes for the tempered fractional diffusion equations
}}

\author{Can Li$^{a,b,}$\footnote{E-mail addresses: mathlican@xaut.edu.cn.},~   Weihua Deng$^{c,}$\footnote{E-mail addresses: dengwh@lzu.edu.cn.}\\
\small{$^a$ Department of Applied Mathematics, School of Sciences,  Xi'an University of Technology, }\\
        \small{ Xi'an, Shaanxi 710054, P.R. China.}\\
\small{$^{b}$ Beijing Computational Science Research Center, Beijing 100084, P.R. China.}\\
\small{$^{c}$ School of Mathematics and Statistics,
Gansu Key Laboratory of Applied Mathematics and Complex Systems,}\\
\small{ Lanzhou University, Lanzhou 730000, P.R. China.}\\
        \vspace{-0.1cm}
        }

\date{}
\begin{document}
\maketitle \makeatletter
\newcommand{\rmnum}[1]{\romannumeral #1} 　　
\newcommand{\Rmnum}[1]{\expandafter\@slowromancap\romannumeral #1@}
\makeatother \vspace{-1cm}
\begin{abstract}
L\'{e}vy flight models whose jumps have infinite moments are mathematically used to describe the superdiffusion
in complex systems. Exponentially tempering L\'{e}vy measure of L\'{e}vy flights leads to the tempered stable L\'{e}vy processes which combine both the $\alpha$-stable and Gaussian trends; and the very large jumps are unlikely and all their moments exist. The probability density functions of the tempered stable L\'{e}vy processes solve the tempered fractional diffusion equation. This paper focuses on designing the high order difference schemes for the tempered fractional diffusion equation on bounded domain. The high order difference approximations, called the tempered and weighted and shifted Gr\"{u}nwald difference (tempered-WSGD) operators, in space are obtained by using the properties of the tempered fractional calculus and weighting and shifting their first order Gr\"{u}nwald type difference approximations. And the Crank-Nicolson discretization is used in the time direction. The stability and convergence of the presented numerical schemes are established; and the numerical experiments are performed to confirm the theoretical results and testify the effectiveness of the schemes.

\end{abstract}



\section{Introduction}
\setlength{\unitlength}{1cm}
The probability density function of L\'{e}vy flights \cite{Metzler:04, Meerschaert:09} has a characteristic
function $e^{-D_{\alpha}|k|^{\alpha}t}\,(0<\alpha<2)$ of stretched Gaussian form, causing the asymptotic decay as
$ |x|^{-1-\alpha}$. It produces that the second moment diverges, i.e., $\langle x^{2}(t)\rangle=\infty$. The divergent
second moments may not be feasible for some even non-Brownian physical processes of practical interest which take place in bounded domains and involve observables with finite moments.
To overcome the divergence of the variance, many techniques are adopted. By simply discarding the very large jumps, Mantegna and Stanley \cite{Mantegna:94} introduce the truncated L\'{e}vy flights and show that the obtained stochastic process ultraslowly converges to a Gaussian. From the point of view of an experimental study, because of the limited time, no expected Gaussian behavior can be observed.
Two other modifications to achieve finite second
moments are proposed by Sokolov et al \cite{Sokolov:89}, who add a high-order power-law factor, and Chechkin et al. \cite{Chechkin:05}, who
add a nonlinear friction term. Exponentially tempering the probability of large jumps of L\'{e}vy flights, i.e., making the L\'{e}vy density decay as $ |x|^{-1-\alpha}\mathrm{e}^{-\lambda |x|}$ with $\lambda>0$, to get finite moments seems to be the most popular one; and the corresponding tempered fractional differential equations are derived \cite{Baeumera:10, Cartea:07a, Cartea:07, Cartea:09}. In fact, the power-law waiting time can also be exponentially tempered \cite{Meerschaert:09}. This paper focused on providing high order schemes for numerically solving tempered fractional diffusion equations, which involve tempered fractional derivatives. In fact, the tempered fractional integral has a long history.
Buschman's earlier work \cite{Buschman:72} reports the fractional integration with weak
singular and exponential kernels; for more detailed discussions, see  Srivastava and  Buschman's book \cite{Srivastava:77} and the references therein. The definitions of tempered fractional calculus are much similar to the ones of fractional substantial calculus \cite{Carmi:10}; but they are introduced from completely different physical backgrounds; e.g., fractional substantial calculus is used to characterize the functional distribution of anomalous diffusion \cite{Carmi:10}. Mathematically the fractional substantial calculus is time-space coupled operator but the tempered fractional calculus is not; numerically the fractional substantial calculus is discretized in the time direction \cite{Deng:14}, but here the tempered fractional calculus is treated as space operator.

Tempered fractional calculus is the generalization of fractional calculus, or fractional calculus is the special/limiting case of tempered fractional calculus.  Some important progresses have been made for numerically solving the fractional partial differential equations (PDEs), e.g., the finite difference methods are used to simulate the space fractional advection diffusion equations \cite{Lynch:03,Meerschaert:04,Liu:04}. Recently, it seems that more efforts of the researchers are put on the high order schemes and fast algorithms. Based on the Toeplitz-like structure of the matrix corresponding the finite difference methods of fractional PDEs, Wang et al \cite{wang:10} numerically solve the fractional diffusion equations with the $N\log ^2N$ computational cost. Later, Pang and Sun \cite{Pang:12} propose a multigrid method to solve the discretized system of the fractional diffusion equation. By introducing the linear spline approximation, Sousa and Li present a second
order discretization for the Riemann-Liouville fractional derivatives, and establish an
unconditionally stable weighted finite difference method for the one-dimensional
fractional diffusion equation in \cite{Sousaa:11}. Ortigueira \cite{Ortigueira:06} gives the ``fractional centred derivative"
to approximate the Riesz fractional derivative with second order accuracy; and this
method is used by \c{C}elik and Duman in \cite{Celik:11} to approximate fractional
diffusion equation with the Riesz fractional derivative in a finite domain.
More recently, by weighting and shifting the Gr\"{u}nwald discretizations, Tian et al \cite{Tian:11} propose a class of second order difference approximations, called WSGD operators, to the Riemann-Liouville fractional derivatives.

So far, there are limited works addressing the finite difference schemes for the tempered fractional
diffusion equations. Baeumera and Meerschaert \cite{Baeumera:10} provide finite difference and particle tracking methods for solving the tempered fractional diffusion equation with the second order accuracy. The stability and convergence of the provided schemes are discussed.
Cartea and del-Castillo-Negrete \cite{Cartea:07a} derive a general finite difference scheme to numerically solve a Black-Merton-Scholes model with tempered fractional derivatives. Recently, Marom and Momoniat \cite{Momoniat:09} compare the numerical solutions of three kinds of fractional Black-Merton-Scholes equations with tempered fractional derivatives. And the
stability and convergence of the presented schemes are not given. To the best of our knowledge, there is no published work to the high order difference schemes for the tempered fractional diffusion equation.  In this paper, with the similar method presented in \cite{Meerschaert:04,Baeumera:10}, we first propose the first order shifted Gr\"{u}nwald type approximation for the tempered fractional calculus; then motivated by the idea in \cite{Tian:11}, we design a series of high order schemes, called the tempered-WSGD operators, by weighting and shifting the first order Gr\"{u}nwald type approximations to the tempered fractional calculus. The obtained high order schemes are applied to solve the tempered fractional diffusion equation and the Crank-Nicolson discretization is used in the time direction. The unconditionally numerical stability and convergence are detailedly discussed; and the corresponding numerical experiments are carried out to illustrate the effectiveness of the schemes.

The remainder of the paper is organized as follows. In Sec. \ref{sec:2}, we introduce the definitions of the tempered fractional calculus and derive their first order shifted Gr\"{u}nwald type approximations and the high order discretizations, the tempered-WSGD operators. In Sec. \ref{sec:3}, the tempered fractional diffusion equation is numerically solved by using the tempered-WSGD operators to approximate the space derivative and the Crank-Nicolson discretization to the time derivative; and the numerical stability and convergence are discussed.  The effectiveness and convergence orders of the presented schemes are numerically verified in Sec. \ref{sec:4}. And the concluding remarks are given in the last section.

\newpage
\section{Definitions of the tempered fractional calculus and the derivation of the tempered-WSGD operators}\label{sec:2}
We first introduce the definitions of the tempered fractional integral and derivative then focus on deriving their high order discretizations, the tempered-WSGD operators.

\subsection{Definitions and Fourier transforms of the tempered fractional calculus}
We introduce the definitions of the tempered fractional calculus and perform their Fourier transforms.

\begin{definition}[\cite{Buschman:72,Cartea:07}]\label{Def1}
Let $u(x)$ be piecewise continuous on $(a,\infty)$ $($or $(-\infty,b)$ corresponding to the right integral$)$ and integrable on any finite subinterval of $[a,\infty)$  $($or $(-\infty,b]$ corresponding to the right integral$)$,  $\sigma>0$, $\lambda\geq0$. Then
\begin{itemize}
  \item[\textrm{(1)}]the left Riemann-Liouville tempered fractional integral of order $\sigma$ is defined to be
    \begin{equation*}
      _{a}D_x^{-\sigma,\lambda}u(x)=\frac{1}{\Gamma(\sigma)}\int_{a}^xe^{-\lambda( x-\xi)}(x-\xi)^{\sigma-1}u(\xi)\mathrm{d}\xi;
    \end{equation*}
  \item[\textrm{(2)}]the right Riemann-Liouville tempered fractional integral of order $\sigma$ is defined to be
    \begin{equation*}
      _xD_{b}^{-\sigma,\lambda}u(x)=\frac{1}{\Gamma(\sigma)}\int_x^b e^{-\lambda( \xi-x)}(\xi-x)^{\sigma-1}u(\xi)\mathrm{d}\xi.
    \end{equation*}
\end{itemize}
\end{definition}
\begin{definition}[\cite{Oldham:74,Samko:93,Podlubny:99}]\label{Def2}
 For $\alpha \in (n-1,n),n\in \mathbb{N}^{+}$, let $u(x)$ be $(n-1)$-times continuously differentiable on $(a,\infty)$  $($or $(-\infty,b)$ corresponding to the right derivative$)$ and its $n$-times derivative be integrable on any subinterval of $[a,\infty)$ $($or $(-\infty,b]$ corresponding to the right derivative$)$. Then
\begin{itemize}
  \item[\textrm{(1)}] the left Riemann-Liouville fractional derivative:
    \begin{equation*}
      _{a}D_x^{\alpha}u(x)=\frac{1}{\Gamma(n-\alpha)}\frac{\mathrm{d}^n}{\mathrm{d}
          x^n}\int_{a}^x\frac{u(\xi)}{(x-\xi)^{\alpha-n+1}}\mathrm{d}\xi;
    \end{equation*}
  \item[\textrm{(2)}] the right Riemann-Liouville fractional derivative:
    \begin{equation*}
      _xD_{b}^{\alpha}u(x)=\frac{(-1)^n}{\Gamma(n-\alpha)}\frac{\mathrm{d}^n}{\mathrm{d}
          x^n}\int_x^b\frac{u(\xi)}{(\xi-x)^{\alpha-n+1}}\mathrm{d}\xi.
    \end{equation*}
\end{itemize}
\end{definition}
\begin{definition}[\cite{Baeumera:10,Cartea:07}]\label{def33}
For $\alpha \in (n-1,n),n\in \mathbb{N}^{+}$, let $u(x)$ be $(n-1)$-times continuously differentiable on $(a,\infty)$  $($or $(-\infty,b)$ corresponding to the right derivative$)$ and its $n$-times derivative be integrable on any subinterval of $[a,\infty)$  $($or $(-\infty,b]$ corresponding to the right derivative$)$, $\lambda\ge0$. Then
\begin{itemize}
  \item[\textrm{(1)}] the left Riemann-Liouville tempered  fractional derivative:
    \begin{equation*}
      _{a}D_x^{\alpha,\lambda}u(x)
      =e^{-\lambda x}{{_{a}}D_x^{\alpha}}\big( e^{\lambda x}u(x)\big)=\frac{e^{-\lambda x}}{\Gamma(n-\alpha)}\frac{\mathrm{d}^n}{\mathrm{d}
          x^n}\int_{a}^x\frac{e^{\lambda \xi}u(\xi)}{(x-\xi)^{\alpha-n+1}}\mathrm{d}\xi;
    \end{equation*}
  \item[\textrm{(2)}] the right Riemann-Liouville tempered  fractional derivative:
    \begin{equation*}
      _xD_{b}^{\alpha,\lambda}u(x)
      =e^{\lambda x}{{_x}D_{b}^{\alpha}}\big( e^{-\lambda x}u(x)\big)=\frac{(-1)^ne^{\lambda x}}{\Gamma(n-\alpha)}\frac{\mathrm{d}^n}{\mathrm{d}
          x^n}\int_x^{b}\frac{e^{-\lambda \xi}u(\xi)}{(\xi-x)^{\alpha-n+1}}\mathrm{d}\xi.
    \end{equation*}
\end{itemize}
If $\lambda=0$, then the left and right Riemann-Liouville tempered fractional derivatives
 $_{a}D_x^{\alpha,\lambda}u(x)$ and $_xD_{b}^{\alpha,\lambda}u(x)$ reduce to
 the left and right Riemann-Liouville fractional derivatives $_{a}D_x^{\alpha}u(x)$
 and $_xD_{b}^{\alpha}u(x)$ defined in Definition \ref{Def2}.
\end{definition}
\begin{definition}\label{def44}
The variants of the left and right Riemann-Liouville tempered fractional derivatives are defined as \cite{Baeumera:10,Cartea:07,Meerschaert:12}
\begin{equation}\label{Lrl}
_{a}\mathbf{D}_x^{\alpha,\lambda}u(x)=
\begin{cases}
\displaystyle
_{a}D_x^{\alpha,\lambda}u(x)-\lambda^{\alpha}u(x),& \text{$0<\alpha<1$},\\
\displaystyle
_{a}D_x^{\alpha,\lambda}u(x)-\alpha\lambda^{\alpha-1}\partial_{x}u(x)-\lambda^{\alpha}u(x),& \text{$1<\alpha<2$};
\end{cases}
\end{equation}
and
\begin{equation}\label{Rrl}
_x\mathbf{D}_{b}^{\alpha,\lambda}u(x)=
\begin{cases}
\displaystyle
_xD_{b}^{\alpha,\lambda}u(x)-\lambda^{\alpha}u(x),& \text{$0<\alpha<1$},\\
\displaystyle
_xD_{b}^{\alpha,\lambda}u(x)+\alpha\lambda^{\alpha-1}\partial_{x}u(x)-\lambda^{\alpha}u(x),& \text{$1<\alpha<2$},
\end{cases}
\end{equation}
where $\partial_{x}$ denotes the classic first order derivative $\frac{\partial}{\partial x}$.
\end{definition}
\begin{remark}
In the above definitions, the `a' can be extended to `$-\infty$' and `b' to `$+\infty$'. In the following analysis, we assume that $u(x)$ is defined on $[a,b]$ and whenever necessary $u(x)$ can be smoothly zero extended to $(-\infty, b)$ or $(a,+\infty)$ or even $(-\infty,+\infty)$. Then  ${_{-\infty}}D_x^{-\sigma,\lambda}u(x)={_{a}}D_x^{-\sigma,\lambda}u(x)$;
${_x}D_{+\infty}^{-\sigma,\lambda}u(x)={_x}D_{b}^{-\sigma,\lambda}u(x)$; ${_{-\infty}}D_x^{\alpha,\lambda}u(x)={_{a}}D_x^{\alpha,\lambda}u(x)$; and ${_x}D_{+\infty}^{\alpha,\lambda}u(x)={_x}D_{b}^{\alpha,\lambda}u(x)$.
\end{remark}
\begin{lemma}[\cite{Buschman:72,Srivastava:77,Baeumera:10}]\label{lem0a}
  Let $u(x)$ and its $n$-times derivative belong to $L^q(\mathbb{R})$, $q \geq 1$.
 Then the Fourier transforms of the left and right Riemann-Liouville tempered fractional integrals are
  \begin{equation}\label{ffb}
     \mathcal{F}({_{-\infty}}D_x^{-\sigma,\lambda}u(x))=(\lambda+i\omega)^{-\sigma}\hat{u}(\omega);
  \end{equation}
  and
  \begin{equation}\label{ffc}
     \mathcal{F}({_x}D_{+\infty}^{-\sigma,\lambda}u(x))=(\lambda-i\omega)^{-\sigma}\hat{u}(\omega)
  \end{equation}
and the Fourier transforms of the left and right Riemann-Liouville tempered fractional derivatives are
  \begin{equation}\label{ffbv}
     \mathcal{F}({_{-\infty}}D_x^{\alpha,\lambda}u(x))=(\lambda+i\omega)^{\alpha}\hat{u}(\omega);
  \end{equation}
  and
  \begin{equation}\label{ffcv}
     \mathcal{F}({_x}D_{+\infty}^{\alpha,\lambda}u(x))=(\lambda-i\omega)^{\alpha}\hat{u}(\omega)
  \end{equation}
and the Fourier transforms of the variants of the left and right Riemann-Liouville tempered fractional derivatives give
\begin{equation}\label{Lrl}
\mathcal{F}(_{-\infty}\mathbf{D}_x^{\alpha,\lambda}u(x))=
\begin{cases}
\displaystyle
(\lambda+i\omega)^{\alpha}\hat{u}(\omega)-\lambda^{\alpha}\hat{u}(\omega),& \text{$0<\alpha<1$},\\
\displaystyle
(\lambda+i\omega)^{\alpha}\hat{u}(\omega)-\alpha i\omega\lambda^{\alpha-1}\hat{u}(\omega)-\lambda^{\alpha}\hat{u}(\omega),& \text{$1<\alpha<2$};
\end{cases}
\end{equation}
and
\begin{equation}\label{Rrl}
\mathcal{F}(_x\mathbf{D}_{+\infty}^{\alpha,\lambda}u(x))=
\begin{cases}
\displaystyle
(\lambda-i\omega)^{\alpha}\hat{u}(\omega)-\lambda^{\alpha}\hat{u}(\omega),& \text{$0<\alpha<1$},\\
\displaystyle
(\lambda-i\omega)^{\alpha}\hat{u}(\omega)+\alpha i\omega\lambda^{\alpha-1}\hat{u}(\omega)-\lambda^{\alpha}\hat{u}(\omega),& \text{$1<\alpha<2$},
\end{cases}
\end{equation}
where the Fourier transform of $u$ is defined by
  \begin{equation*}
   \mathcal{F}(u(x))(\omega)=\int_{\mathbb{R}}\mathrm{e}^{-i\omega x}u(x)\mathrm{d}x,\,i^2=-1.
  \end{equation*}
\end{lemma}
\begin{remark} [\cite{Carmi:10,Deng:14}]
The left and right Riemann-Liouville tempered fractional derivatives can be, respectively, rewritten as
    \begin{equation*}
      _{-\infty}D_x^{\alpha,\lambda}u(x)
      =\frac{1}{\Gamma(n-\alpha)}\left(\frac{\mathrm{d}}{\mathrm{d}
          x}+ \lambda\right)^n\int_{-\infty}^x\frac{e^{-\lambda(x-\xi)}u(\xi)}{(x-\xi)^{\alpha-n+1}}\mathrm{d}\xi
          ;
    \end{equation*}
and
    \begin{equation*}
      _xD_{+\infty}^{\alpha,\lambda}u(x)
      =\frac{(-1)^n}{\Gamma(n-\alpha)}\left(\frac{\mathrm{d}}{\mathrm{d}
          x}- \lambda\right)^n\int_x^{+\infty}\frac{e^{-\lambda(\xi-x)}u(\xi)}{(\xi-x)^{\alpha-n+1}}\mathrm{d}\xi.
    \end{equation*}
\end{remark}

\subsection{Discretizations of the tempered fractional calculus}
In this subsection, we derive the Gr\"{u}nwald type discretizations for the tempered fractional calculus. The standard Gr\"{u}nwald discretization  generally yields an unstable finite difference scheme when it is used to solve the time dependent fractional PDEs \cite{Meerschaert:04}.
To remedy this defect, Meerschaert et al introduce a shifted Gr\"{u}nwald formula. The similar numerical unstability also happens for the time dependent tempered fractional PDEs; so the shift for the Gr\"{u}nwald type discretizations of the tempered fractional derivative is also necessary.

\begin{lemma}\label{lem1}
Let $u(x)\in L^1(\mathbb{R})$, ${_{-
  \infty}}D_x^{\alpha+1,\lambda}u$ and its Fourier transform belong to $L^1(\mathbb{R})$; $p\in\mathbb{R},h>0,\,\lambda\geq0$ and $\alpha\in(n-1,n),n\in \mathbb{N}^{+}$. Defining the shifted Gr\"{u}nwald type difference operator
\begin{equation}\label{lem01}
A^{\alpha,\lambda}_{h,p}u(x):=\frac{1}{h^{\alpha}}\sum_{k=0}^{+\infty}w^{(\alpha)}_{k}e^{-(k-p)h\lambda}u(x-(k
-p)h)-\frac{1}{h^{\alpha}}\big(e^{ph\lambda}(1-e^{-h\lambda})^{\alpha}\big)u(x),
\end{equation}
then
\begin{equation}\label{lem02}
A^{\alpha,\lambda}_{h,p}u(x)={_{-\infty}}D_x^{\alpha,\lambda}u(x)-\lambda^\alpha u(x)+O(h),
\end{equation}
where $w_k^{(\alpha)}=(-1)^k\binom{\alpha}{k},\,k\geq0$  denotes the normalized Gr\"{u}nwald weights.
\end{lemma}
\begin{remark} \label{remark3}
The point $x+(p-\alpha/2)h$ is the superconvergent point of the
approximation $A^{\alpha,\lambda}_{h,p}$ to
${_{-\infty}}D_y^{\alpha,\lambda}-\lambda^\alpha$, i.e.,
$A^{\alpha,\lambda}_{h,p}u(x)={_{-\infty}}D_y^{\alpha,\lambda}u(y)- \lambda^{\alpha}u(y)+O(h^2)$
with $y=x+(p-\alpha/2)h$ (the deriving process is similar to the one
given in \cite{Nasir:13}).
\end{remark}
\begin{remark}
Under the assumption given in Lemma \ref{lem1}, for tempered fractional derivatives defined in \eqref{Lrl}, we have \cite{Baeumera:10}
\begin{equation}\label{lem02remark}
A^{\alpha,\lambda}_{h,p}u(x)=
\begin{cases}
\displaystyle {_{-\infty}}\mathbf{D}_x^{\alpha,\lambda}u(x)+O(h),& \text{$0<\alpha<1$},\\
\displaystyle
{_{-\infty}}\mathbf{D}_x^{\alpha,\lambda}u(x)+\alpha \lambda^{\alpha-1}\partial_xu(x)+O(h),& \text{$1<\alpha<2$}.
\end{cases}
\end{equation}
\end{remark}
\begin{proof}
The proof is similar to the one given in \cite{Baeumera:10}. Taking Fourier transform on both sides of (\ref{lem01}), we obtain
  \begin{equation}
    \begin{aligned}\label{lem03}
     \mathcal{F}[A^{\alpha,\lambda}_{h,p}u](\omega) & =\frac{1}{h^\alpha}
      \sum_{k=0}^{
+\infty}w_k^{(\alpha)}e^{-(k-p)h(\lambda+i\omega)}
      \hat{u}(\omega)-\frac{1}{h^{\alpha}}\big(e^{ph\lambda}(1-e^{-h\lambda})^{\alpha}\big)\hat{u}(\omega) \\
      & =e^{ph(\lambda+i\omega)}\bigg(\frac{1-e^{-h(\lambda+i\omega)}}{h}\bigg)^\alpha
        \hat{u}(\omega)-e^{ph\lambda}\bigg(\frac{1-e^{-h\lambda}}{h}\bigg)^{\alpha}\hat{u}(\omega)
        \\
      & =\big[(\lambda+i\omega)^{\alpha}P_h
      (\lambda+i\omega)-\lambda^{\alpha}P_h
      (\lambda)\big]\hat{u}(\omega),
    \end{aligned}
  \end{equation}
  where
  \begin{equation}\label{lem04ap}
    P_h(z)=e^{phz}\bigg(\frac{1-e^{-hz}}{hz}\bigg)^\alpha=1+(p-\frac{\alpha}{2})hz+O(|z|^2),\,\, \mbox{with}\,\,
    z=\lambda+i\omega\,\,\textrm{or}\,\,\lambda.
  \end{equation}
  Denoting $$\hat{\phi}(\omega,h)=\mathcal{F}[A^{\alpha,\lambda}_{h,p}u](\omega)-\mathcal{F}[{_{-\infty}}D_x^{\alpha,\lambda}u- \lambda^{\alpha}u](\omega)=
  \big[(\lambda+i\omega)^{\alpha}\big(P_h(\lambda+i\omega)-1\big)
  -\lambda^{\alpha}\big(P_h(\lambda)-1\big)\big]\hat{u}(\omega),$$ from (\ref{lem03}) and (\ref{ffbv})
  there exists
  \begin{equation*}
    |\hat{\phi}(\omega,h)|\le C\big[h|(\lambda+i\omega)|^{\alpha+1}+h|\lambda|^{\alpha+1}\big]|\hat{u}(\omega
    )|.
  \end{equation*}
  With the condition $\mathcal{F}[{_{-\infty}}D_x^{\alpha+1,\lambda}u](k)\in L^1(\mathbb{R})$, and using the Riemann-Lebesgue Lemma, it yields
  \begin{equation*}
    \begin{aligned}
     |A^{\alpha,\lambda}_{h,p}u(x)-~_{-\infty}D_{x}^{\alpha,\lambda}u(x)+\lambda^{\alpha}u(x)| & =|\phi|\le
    \frac{1}{2\pi}\int_{\mathbb{R}}|\hat{\phi}(\omega,h)|d\omega \\
      &\le C\|\mathcal{F}[{_{-\infty}}D_x^{\alpha+1,\lambda}u+\lambda^{\alpha+1}u(x)](\omega)\|_{L^1}h=O(h),
    \end{aligned}
  \end{equation*}
where the property of the Fourier transforms for the left Riemann-Liouville tempered fractional derivatives (\ref{ffbv}) is used.
\end{proof}
\begin{lemma}\label{lem2}
Let $u(x)\in L^1(\mathbb{R})$, ${_{x}}D_{+\infty}^{\alpha+1,\lambda}u$ and its Fourier transform belong to $L^1(\mathbb{R})$; $p\in\mathbb{R},\,h>0,\,\lambda\geq0$ and $\alpha\in(n-1,n),n\in \mathbb{N}^{+}$. Define the tempered shifted Gr\"{u}nwald type difference operator
\begin{equation}\label{lem20}
B^{\alpha,\lambda}_{h,p}u(x):=\frac{1}{h^{\alpha}}\sum_{k=0}^{+\infty}w^{(\alpha)}_{k}e^{-(k-p)h\lambda}u(x+(k
-p)h)-\frac{1}{h^{\alpha}}\big(e^{ph\lambda}(1-e^{-h\lambda})^{\alpha}\big)u(x).
\end{equation}
Then
\begin{equation}\label{lem201}
B^{\alpha,\lambda}_{h,q}u(x)=
 {_x}D_{+\infty}^{\alpha,\lambda}u(x)-\lambda^\alpha u(x)+O(h).
\end{equation}
\end{lemma}
\begin{remark}\label{remark4}
The point $x-(p-\alpha/2)h$ is the superconvergent point of the
approximation $B^{\alpha,\lambda}_{h,p}$ to
${_{y}}D_\infty^{\alpha,\lambda}-\lambda^{\alpha}$, i.e.,
$B^{\alpha,\lambda}_{h,p}u(x)={_{y}}D_\infty^{\alpha,\lambda}u(y)-\lambda^{\alpha}u(y)+O(h^2)$
with $y=x-(p-\alpha/2)h$ (the deriving process is similar to the one
given in \cite{Nasir:13}).
\end{remark}
\begin{remark}
Under the assumption given in Lemma \ref{lem2}, for tempered fractional derivatives defined in  \eqref{Rrl}, we have
\begin{equation}\label{lem201remark}
B^{\alpha,\lambda}_{h,q}u(x)=
\begin{cases}
\displaystyle {_x}\mathbf{D}_{+\infty}^{\alpha,\lambda}u(x)+O(h),& \text{$0<\alpha<1$},\\
\displaystyle
{_x}\mathbf{D}_{+\infty}^{\alpha,\lambda}u(x)-\alpha \lambda^{\alpha-1}\partial_xu(x)+O(h),& \text{$1<\alpha<2$}.
\end{cases}
\end{equation}
\end{remark}
\begin{proof}
Taking Fourier transform on both sides of (\ref{lem20}), we obtain
  \begin{equation*}
    \begin{aligned}\label{eq:2.12}
      \mathcal{F}[B^{\alpha,\lambda}_{h,p}u](\omega) & =\frac{1}{h^\alpha}
      \sum_{k=0}^{+\infty}w_k^{(\alpha)}e^{-(k-p)h(\lambda-i\omega)}
      \hat{u}(\omega)-\frac{1}{h^{\alpha}}\big(e^{ph\lambda}(1-e^{-h\lambda})^{\alpha}\big)\hat{u}(\omega) \\
      & = e^{ph(\lambda-i\omega)}\bigg(\frac{1-e^{-h(\lambda-i\omega)}}{h}\bigg)^\alpha
        \hat{u}(\omega)-e^{ph\lambda}\bigg(\frac{1-e^{-h\lambda}}{h}\bigg)^{\alpha}\hat{u}(\omega)
        \\
      & =\big[(\lambda-i\omega)^{\alpha}P_h
      (\lambda-i\omega)-\lambda^{\alpha}P_h
      (\lambda)\big]\hat{u}(\omega),
    \end{aligned}
  \end{equation*}
  where $P_h(z)$ is defined by (\ref{lem04ap}) with $z=\lambda-i\omega$ or $\lambda$.
  Denoting $\hat{\phi}(\omega,h)=\mathcal{F}[B^{\alpha,\lambda}_{h,p}u](\omega)-\mathcal{F}[{_{x}}D_{\infty}^{\alpha,\lambda}u- \lambda^{\alpha}u](\omega)$, then with the similar method used in the proof of Lemma \ref{lem1}, and using the Fourier transform of the right Riemann-Liouville tempered fractional derivative (\ref{ffcv}), we obtain
  \begin{equation*}
    \begin{aligned}
     |B^{\alpha,\lambda}_{h,p}u(x)-{_x}D_{+\infty}^{\alpha,\lambda}u(x)+ \lambda^{\alpha}u(x)|
      & =|\phi|\le
    \frac{1}{2\pi}\int_{\mathbb{R}}|\hat{\phi}(\omega,h)|d\omega \\
      &\le C\|\mathcal{F}[{_{x}}D_{+\infty}^{\alpha+1,\lambda}u+\lambda^{\alpha+1}u(x)](\omega)\|_{L^1}h=O(h).
    \end{aligned}
  \end{equation*}
\end{proof}

The approximation accuracy of the classic difference operator can be improved by adding the band of discretization stencils  \cite{Fornberg:98}. And then the computational cost increases accordingly. However, because of the nonlocal property of the fractional operator, even for the first order discretizations, the stencil covers the whole interval. Without introducing new computational cost, we can improve the approximation accuracy of the discretized fractional operators by modifying the Gr\"{u}nwald type weights. The improved discretized tempered fractional operators are called tempered weighted and shifted Gr\"{u}nwald difference (tempered-WSGD) operators. 

\begin{theorem}\label{thm1}
  Let $u(x)\in L^1(\mathbb{R})$, ${_{-
  \infty}}D_x^{\alpha+\ell,\lambda}u$ and its Fourier transform belong to $L^1(\mathbb{R})$;
 and define the left tempered-WSGD operator by
  \begin{equation}\label{rWSGD}
  _{_L}\mathcal{D}_{h,p_{1},p_{2},\ldots,p_{m}}^{\alpha,\gamma_{1},\gamma_{2},\ldots,\gamma_{m}}u(x)
  =\sum_{j=1}^{m}\gamma_{j}A_{h,p_{j}}^{\alpha,\lambda}u(x),
  \end{equation}
where $p_j$ and $\gamma_j$ are determined by (\ref{system0})-(\ref{system3}). Then, for any integer $m\geq\ell$, there exists
\begin{equation}\label{thm01}
 _{_L}\mathcal{D}_{h,p_{1},p_{2},\ldots,p_{m}}^{\alpha,\gamma_{1},\gamma_{2},\ldots,\gamma_{m}}u(x)
 =
{_{-\infty}}D_x^{\alpha,\lambda}u(x)- \lambda^{\alpha}u(x)+O(h^\ell),
\end{equation}
  uniformly for $x\in\mathbb{R}$.

Let $u(x)\in L^1(\mathbb{R})$, $_{x}D_{\infty}^{\alpha+\ell,\lambda}u$ and its Fourier transform belong to $L^1(\mathbb{R})$;
 and define the right tempered-WSGD operator by
   \begin{equation}\label{lWSGD}
   {_R}\mathcal{D}_{h,p_{1},p_{2},\ldots,p_{m}}^{\alpha,\gamma_{1},\gamma_{2},\ldots,\gamma_{m}}u(x)
   =\sum_{j=1}^{m}\gamma_{j}B_{h,p_{j}}^{\alpha,\lambda}u(x),
   \end{equation}
   where $p_j$ and $\gamma_j$ are determined by (\ref{system0})-(\ref{system3}). Then, for any integer $m\geq\ell$, there is
  \begin{equation}\label{thm002}
    {_R}\mathcal{D}_{h,p_{1},p_{2},\ldots,p_{m}}^{\alpha,\gamma_{1},\gamma_{2},\ldots,\gamma_{m}}u(x)
    =
{_{x}}D_{+\infty}^{\alpha,\lambda}u(x)- \lambda^{\alpha}u(x)+O(h^\ell),
  \end{equation}
  uniformly for $x\in\mathbb{R}$.
  \begin{itemize}
  \item[]For $\ell=2$, $p_{j},\gamma_{j}$ are real numbers and satisfy the linear system
  \begin{equation}\label{system0}
\begin{cases}
\displaystyle
\sum_{j=1}^{m}\gamma_{j}=1,\\
\displaystyle
\sum_{j=1}^{m}\gamma_{j}\bigg[p_{j}-\frac{\alpha}{2}\bigg]=0.
\end{cases}
\end{equation}
  \item[]For $\ell=3$, $p_{j},\gamma_{j}$ are real numbers and satisfy
  \begin{equation}\label{system1}
\begin{cases}
\displaystyle
\sum_{j=1}^{m}\gamma_{j}=1,\\
\displaystyle
\sum_{j=1}^{m}\gamma_{j}\bigg[p_{j}-\frac{\alpha}{2}\bigg]=0,\\
\displaystyle
\sum_{j=1}^{m}\gamma_{j}\bigg[\frac{p_{j}^2}{2}-\frac{\alpha p_{j}}{2}+\frac{\alpha}{6}+\frac{\alpha(\alpha-1)}{8}\bigg]=0.
\end{cases}
\end{equation}
\item[]For $\ell=4$, $p_{j},\gamma_{j}$ are real numbers and the following hold
  \begin{equation}\label{system2}
\begin{cases}
\displaystyle
\sum_{j=1}^{m}\gamma_{j}=1,\\
\displaystyle
\sum_{j=1}^{m}\gamma_{j}\bigg[p_{j}-\frac{\alpha}{2}\bigg]=0,\\
\displaystyle
\sum_{j=1}^{m}\gamma_{j}\bigg[\frac{p_{j}^2}{2}-\frac{\alpha p_{j}}{2}+\frac{\alpha}{6}+\frac{\alpha(\alpha-1)}{8}\bigg]=0,\\
\displaystyle
\sum_{j=1}^{m}\gamma_{j}\bigg[\frac{p_{j}^3}{6}-\frac{\alpha p_{j}^2}{4}+\big(\frac{\alpha}{6}+\frac{\alpha(\alpha-1)}{8}\big)p_{j}
      -\frac{\alpha}{24}-\frac{\alpha(\alpha-1)}{12}-\frac{\alpha(\alpha-1)(\alpha-2)}{48}\bigg]=0.
\end{cases}
\end{equation}
\item[]For $\ell=5$, $p_{j},\gamma_{j}$ are real numbers and the following hold
  \begin{equation}\label{system3}
\begin{cases}
\displaystyle
\sum_{j=1}^{m}\gamma_{j}=1,\\
\displaystyle
\sum_{j=1}^{m}\gamma_{j}\bigg[p_{j}-\frac{\alpha}{2}\bigg]=0,\\
\displaystyle
\sum_{j=1}^{m}\gamma_{j}\bigg[\frac{p_{j}^2}{2}-\frac{\alpha p_{j}}{2}+\frac{\alpha}{6}+\frac{\alpha(\alpha-1)}{8}\bigg]=0,\\
\displaystyle
\sum_{j=1}^{m}\gamma_{j}\bigg[\frac{p_{j}^3}{6}-\frac{\alpha p_{j}^2}{4}+\big(\frac{\alpha}{6}+\frac{\alpha(\alpha-1)}{8}\big)p_{j}
      -\frac{\alpha}{24}-\frac{\alpha(\alpha-1)}{12}-\frac{\alpha(\alpha-1)(\alpha-2)}{48}\bigg]=0,\\
\displaystyle
\sum_{j=1}^{m}\gamma_{j} \bigg[\frac{p_{j}^4}{24}-\frac{\alpha p^3_{j}}{4}+\frac{1}{2}\big(\frac{\alpha}{6}+\frac{\alpha(\alpha-1)}{8}\big)p^2_{j}
+\big(-\frac{\alpha}{24}-\frac{\alpha(\alpha-1)}{12}-\frac{\alpha(\alpha-1)(\alpha-2)}{48}\big)p_{j}\\
\displaystyle
+\frac{\alpha}{120}+\frac{5\alpha(\alpha-1)}{144}+\frac{\alpha(\alpha-1)(\alpha-2)}{48}
+\frac{\alpha(\alpha-1)(\alpha-2)(\alpha-3)}{384}\bigg]=0.
\end{cases}
\end{equation}
\end{itemize}
\end{theorem}
\begin{proof}The standard Fourier transforms are again used here.
Performing the Fourier transform on the left hand of (\ref{rWSGD}), we obtain
  \begin{equation}
    \begin{aligned}\label{eq:2.12}
      \mathcal{F}[_{_L}\mathcal{D}_{h,p_{1},p_{2},\ldots,p_{m}}^{\alpha,\gamma_{1},\gamma_{2},\ldots,\gamma_{m}}u(x)](\omega)& =\sum_{j=1}^{m}\gamma_{j}\bigg(\frac{1}{h^\alpha}
      \sum_{k=0}^{\infty}w_k^{(\alpha)}e^{-(k-p_{j})h(\lambda+i\omega)}\hat{u}(\omega) -\frac{1}{h^{\alpha}}\big(e^{p_jh\lambda}(1-e^{-h\lambda})^{\alpha}\big)\hat{u}(\omega) \bigg) \\
      & = \sum_{j=1}^m\big[(\lambda+i\omega)^{\alpha}P_{h,j}
      (\lambda+i\omega)-\lambda^{\alpha}P_{h,j}
      (\lambda)\big]\hat{u}(\omega)\gamma_{j},
    \end{aligned}
  \end{equation}
  where $P_{h,j}(z)=e^{p_{j}hz}\bigg(\frac{1-e^{-hz}}{hz}\bigg)^\alpha,\,z=\lambda+i\omega$ or $\lambda$, $i=\sqrt{-1}.$
  By a simple Taylor's expansion, we get
  \begin{equation}
    \begin{aligned}\label{lem044}
    e^{p_{j}hz}\bigg(\frac{1-e^{-hz}}{hz}\bigg)^\alpha=&1+\bigg[p_{j}-\frac{\alpha}{2}\bigg]hz
    +\bigg[\frac{p_{j}^2}{2}-\frac{\alpha p_{j}}{2}+\frac{\alpha}{6}+\frac{\alpha(\alpha-1)}{8}\bigg](hz)^2~
     \\
      &+\bigg[\frac{p_{j}^3}{6}-\frac{\alpha p_{j}^2}{4}+\big(\frac{\alpha}{6}+\frac{\alpha(\alpha-1)}{8}\big)p_{j}
      -\frac{\alpha}{24}-\frac{\alpha(\alpha-1)}{12}-\frac{\alpha(\alpha-1)(\alpha-2)}{48}\bigg](hz)^3\\
      &+\bigg[\frac{p_{j}^4}{24}-\frac{\alpha p^3_{j}}{4}+\frac{1}{2}\big(\frac{\alpha}{6}+\frac{\alpha(\alpha-1)}{8})\big)p^2_{j}
      +\big(-\frac{\alpha}{24}-\frac{\alpha(\alpha-1)}{12}-\frac{\alpha(\alpha-1)(\alpha-2)}{48}\big)p_{j}\\
      &
      +\frac{\alpha}{120}+\frac{5\alpha(\alpha-1)}{144}+\frac{\alpha(\alpha-1)(\alpha-2)}{48}
      +\frac{\alpha(\alpha-1)(\alpha-2)(\alpha-3)}{384}\bigg](hz)^4
      \\
      &+O(|zh|^5).
    \end{aligned}
  \end{equation}
  Denoting $\mathcal{F}[_{_L}\mathcal{D}_{h,p_{1},p_{2},\ldots,p_{m}}^{\alpha,\gamma_{1},\gamma_{2},\ldots,\gamma_{m}}u(x)](\omega)
  =\mathcal{F}[{_{-\infty}}D_x^{\alpha,\lambda}u-\lambda^{\alpha}u](\omega)+\hat{\phi}(\omega,h)$, in view of (\ref{lem044}), (\ref{ffbv}), and (\ref{system0})-(\ref{system3}), we have
  \begin{equation}
    |\hat{\phi}(\omega,h)|\le Ch^l\big[|\lambda+i\omega|^{\alpha+\ell}+|\lambda|^{\alpha+\ell}\big]|\hat{u}(\omega)|.
  \end{equation}
  Due to $\mathcal{F}[{_{-\infty}}D_x^{\alpha+\ell,\lambda}u](\omega)\in L^1(\mathbb{R})$, there exists
 \begin{equation*}
    \begin{aligned}
     |_{_L}\mathcal{D}_{h,p_{1},p_{2},\ldots,p_{m}}^{\alpha,\gamma_{1},\gamma_{2},\ldots,\gamma_{m}}u
    -{_{-\infty}}D_x^{\alpha,\lambda}u+\lambda^{\alpha}u|
&  =|\phi|\le
    \frac{1}{2\pi}\int_{\mathbb{R}}|\hat{\phi}(\omega,h)|d\omega\\
&\le C\|\mathcal{F}[{_{-\infty}}D_x^{\alpha+\ell,\lambda}u+\lambda^{\alpha+\ell}u](\omega)\|_{L^1}h^\ell=O(h^\ell).
    \end{aligned}
  \end{equation*}
By the similar arguments we can prove (\ref{thm002}).
\end{proof}
\begin{remark}
Under the assumptions given by Theorem \ref{thm1}, for the tempered fractional derivatives defined in \eqref{Lrl} and \eqref{Rrl}, we deduce that
\begin{equation}\label{thm01remark}
 _{_L}\mathcal{D}_{h,p_{1},p_{2},\ldots,p_{m}}^{\alpha,\gamma_{1},\gamma_{2},\ldots,\gamma_{m}}u(x)
 =
 \begin{cases}
\displaystyle {_{-\infty}}\mathbf{D}_x^{\alpha,\lambda}u(x)+O(h^\ell),,& \text{$0<\alpha<1$},\\
\displaystyle
{_{-\infty}}\mathbf{D}_x^{\alpha,\lambda}u(x)+\alpha \lambda^{\alpha-1}\partial_xu(x)+O(h^\ell),& \text{$1<\alpha<2$};
\end{cases}
  \end{equation}
and
\begin{equation}\label{thm002remark}
    {_R}\mathcal{D}_{h,p_{1},p_{2},\ldots,p_{m}}^{\alpha,\gamma_{1},\gamma_{2},\ldots,\gamma_{m}}u(x)
    =
    \begin{cases}
\displaystyle {_{x}}\mathbf{D}_{+\infty}^{\alpha,\lambda}u(x)+O(h^\ell),& \text{$0<\alpha<1$},\\
\displaystyle
{_{x}}\mathbf{D}_{+\infty}^{\alpha,\lambda}u(x)-\alpha \lambda^{\alpha-1}\partial_xu(x)+O(h^\ell),& \text{$1<\alpha<2$}.
\end{cases}
\end{equation}
\end{remark}
\begin{remark}\label{lem1reaa}
If $u(x)\in L^1(\mathbb{R})$, ${_{-
  \infty}}D_x^{\alpha+1,\lambda}u$ and its Fourier transform belong to $L^1(\mathbb{R})$; $p\in\mathbb{R},h>0,\,\lambda\geq0$ and $\alpha\in(n-1,n),n\in \mathbb{N}^{+}$. Defining the shifted Gr\"{u}nwald type difference operator
\begin{equation}\label{lem01rea}
\tilde{A}^{\alpha,\lambda}_{h,p}u(x):=\frac{1}{h^{\alpha}}\sum_{k=0}^{+\infty}w^{(\alpha)}_{k}e^{-(k-p)h\lambda}u(x-(k
-p)h),
\end{equation}
then
\begin{equation}\label{lem02reaa}
\tilde{A}^{\alpha,\lambda}_{h,p}u(x)={_{-\infty}}D_x^{\alpha,\lambda}u(x)+O(h),
\end{equation}
where $w_k^{(\alpha)}=(-1)^k\binom{\alpha}{k},\,k\geq0$  denotes the normalized Gr\"{u}nwald weights.

If $u(x)\in L^1(\mathbb{R})$, ${_{x}}D_{+\infty}^{\alpha+1,\lambda}u$ and its Fourier transform belong to $L^1(\mathbb{R})$; $p\in\mathbb{R},\,h>0,\,\lambda\geq0$ and $\alpha\in(n-1,n),n\in \mathbb{N}^{+}$. Define the shifted Gr\"{u}nwald type difference operator
\begin{equation}\label{lem20aa}
\tilde{B}^{\alpha,\lambda}_{h,p}u(x):=\frac{1}{h^{\alpha}}\sum_{k=0}^{+\infty}w^{(\alpha)}_{k}e^{-(k-p)h\lambda}u(x+(k
-p)h).
\end{equation}
Then
\begin{equation}\label{lem201aa}
\tilde{B}^{\alpha,\lambda}_{h,q}u(x)={_x}D_{+\infty}^{\alpha,\lambda}u(x)+O(h).
\end{equation}

Moreover, if $u(x)\in L^1(\mathbb{R})$, ${_{-
  \infty}}D_x^{\alpha+\ell,\lambda}u$ and its Fourier transform belong to $L^1(\mathbb{R})$;
 and define the left tempered-WSGD operator by
  \begin{equation}\label{rWSGDa}
  _{_L}\tilde{\mathcal{D}}_{h,p_{1},p_{2},\ldots,p_{m}}^{\alpha,\gamma_{1},\gamma_{2},\ldots,\gamma_{m}}u(x)
  =\sum_{j=1}^{m}\gamma_{j}\tilde{A}_{h,p_{j}}^{\alpha,\lambda}u(x),
  \end{equation}
where $p_j$ and $\gamma_j$ are determined by (\ref{system0})-(\ref{system3}). Then, for any integer $m\geq\ell$, there exists
\begin{equation}\label{thm01a}
 _{_L}\tilde{\mathcal{D}}_{h,p_{1},p_{2},\ldots,p_{m}}^{\alpha,\gamma_{1},\gamma_{2},\ldots,\gamma_{m}}u(x)
 =
{_{-\infty}}D_x^{\alpha,\lambda}u(x)+O(h^\ell),
\end{equation}
  uniformly for $x\in\mathbb{R}$.

If $u(x)\in L^1(\mathbb{R})$, $_{x}D_{\infty}^{\alpha+\ell,\lambda}u$ and its Fourier transform belong to $L^1(\mathbb{R})$;
 and define the right tempered-WSGD operator by
   \begin{equation}\label{lWSGDa}
   {_R}\tilde{\mathcal{D}}_{h,p_{1},p_{2},\ldots,p_{m}}^{\alpha,\gamma_{1},\gamma_{2},\ldots,\gamma_{m}}u(x)
   =\sum_{j=1}^{m}\gamma_{j}\tilde{B}_{h,p_{j}}^{\alpha,\lambda}u(x),
   \end{equation}
   where $p_j$ and $\gamma_j$ are determined by (\ref{system0})-(\ref{system3}). Then, for any integer $m\geq\ell$, there is
  \begin{equation}\label{thm002a}
    {_R}\tilde{\mathcal{D}}_{h,p_{1},p_{2},\ldots,p_{m}}^{\alpha,\gamma_{1},\gamma_{2},\ldots,\gamma_{m}}u(x)
    =
{_{x}}D_{+\infty}^{\alpha,\lambda}u(x)+O(h^\ell),
  \end{equation}
  uniformly for $x\in\mathbb{R}$.
\end{remark}
\begin{remark}
To get the discretizations, including the first and high orders, of the left and right Riemann-Liouville tempered fractional integrals of order $\sigma>0$: $_{-\infty}D_x^{-\sigma,\lambda}u(x)$  and  $ _xD_{\infty}^{-\sigma,\lambda}u(x)$, just use $-\sigma$ to replace $\alpha$ existing in the corresponding discretizations of the left and right Riemann-Liouville tempered fractional derivatives of order $\alpha>0$: $_{-\infty}D_x^{\alpha,\lambda}u(x)$  and  $ _xD_{+\infty}^{\alpha,\lambda}u(x)$.
\end{remark}
Considering a well-defined function $u(x)$ on the bounded interval $[a,b]$, the function $u(x)$ can be zero extended for $x<a$ or $x>b$. Then the $\alpha$-th order left and right Riemann-Liouville tempered fractional derivatives of $u(x)$ at point $x$ can be approximated by the tempered-WSGD operators
  \begin{equation}
    \begin{aligned}\label{fwsgl}
       _aD_x^{\alpha,\lambda}u(x)- \lambda^{\alpha} u(x)=&\sum_{j=1}^{m}\frac{\gamma_j}{h^\alpha}\bigg(\sum_{k=0}^{[\frac{x-a}{h}]
      +p_{j}}w_k^{(\alpha)}e^{-(k-p_{j})h\lambda}u(x-(k-p_{j})h)-\big(e^{p_{j}h\lambda}(1-e^{-h\lambda})^{\alpha}\big)u(x)\bigg)\\
      &+O(h^\ell); \\
       _xD_b^{\alpha,\lambda}u(x)- \lambda^{\alpha} u(x)=&\sum_{j=1}^{m}\frac{\gamma_j}{h^\alpha}\bigg(\sum_{k=0}^{[\frac{b-x}{h}]+p_{j}}
      w_k^{(\alpha)}e^{-(k-p_{j})h\lambda}u(x+(k-p_{j})h)-\big(e^{p_{j}h\lambda}(1-e^{-h\lambda})^{\alpha}\big)u(x)\bigg)\\
      &+O(h^\ell),
    \end{aligned}
  \end{equation}
  and the $\sigma$-th order left and right Riemann-Liouville tempered fractional integrals of $u(x)$ at point $x$ can be approximated by the tempered-WSGD operators
  \begin{equation}
    \begin{aligned}\label{fwsgi}
      & _aD_x^{-\sigma,\lambda}u(x)=\sum_{j=1}^{m}\gamma_jh^\sigma\bigg(\sum_{k=0}^{[\frac{x-a}{h}]
      +p_{j}}w_k^{(-\sigma)}e^{-(k-p_{j})h\lambda}u(x-(k-p_{j})h)\bigg)
      +O(h^\ell), \\
      & _xD_b^{-\sigma,\lambda}u(x)=\sum_{j=1}^{m}\gamma_jh^\sigma\bigg(\sum_{k=0}^{[\frac{b-x}{h}]+p_{j}}
      w_k^{(-\sigma)}e^{-(k-p_{j})h\lambda}u(x+(k-p_{j})h)\bigg)
      +O(h^\ell),
    \end{aligned}
  \end{equation}
 where the weight parameters $\gamma_{j}$ are determined by the above linear algebraic systems given in Theorem
 \ref{thm1}.

\begin{remark}\label{rem:3}

  The parameters $[(x-a)/h]+p_{j}$ are the numbers of the points located on the right/left hand of the point $x$ used for evaluating the $\alpha$-th (or $\sigma$-th)  order left/right Riemann-Liouville tempered fractional derivatives (or integrals) at $x$;
  thus, when employing the discretizations (\ref{fwsgl}) (or (\ref{fwsgi})) for approximating non-periodic boundary problems on bounded interval, $p_{j}$ should be chosen satisfying $|p_{j}|\le1$ to ensure that the nodes at which the values of $u$ are needed in (\ref{fwsgl}) (or (\ref{fwsgi})) are within the bounded interval; otherwise, we need to use another methodology to discretize the tempered fractional derivative when $x$ is close to the
  right/left boundary just like classic ones \cite{Fornberg:98}.
\end{remark}
It is easy to check that any one of the linear systems
(\ref{system0})-(\ref{system3}) with $m=\ell$ has an unique
solution. And for $m>l$, using the knowledge of linear algebra, we
know that the system (\ref{system0})-(\ref{system3}) has infinitely
many solutions. As we have discussed in Theorem \ref{thm1}, in
principle the arbitrarily high order difference approximations can
be obtained. For computational purposes, we are more interested
in the schemes with $|p_{j}|\leq1$. And for the easy of presentation
but without loss of the generality, in the following sections, we
focus on the second order difference approximations ($l=2$) of
\eqref{fwsgl} with three to be determined weights
$\gamma_{j},j=1,2,3$  (m=3), i.e.,
\begin{equation}
    \begin{aligned}\label{eq2.7ww1}
       _{_L}\mathcal{D}_{h,1,0,-1}^{\alpha,\gamma_{1},\gamma_{2},\gamma_{3}}u(x)
      &:=\frac{\gamma_1}{h^\alpha}\sum_{k=0}^{[\frac{x-a}{h}]+1}w^{(\alpha)}_{k}e^{-(k-1)h\lambda}u(x-(k-1)h)     +\frac{\gamma_2}{h^\alpha}\sum_{k=0}^{[\frac{x-a}{h}]}w^{(\alpha)}_{k}e^{-kh\lambda}u(x-kh)
      \\&~~+\frac{\gamma_3}{h^\alpha}\sum_{k=0}^{[\frac{x-a}{h}]-1}w^{(\alpha)}_{k}e^{-(k+1)h\lambda}u(x-(k+1)h)\\
      &~~-\frac{1}{h^\alpha}
      \big((\gamma_{1}e^{h\lambda}+\gamma_{2}+\gamma_{3}e^{-h\lambda})(1-e^{-h\lambda})^\alpha\big)u(x);
 \end{aligned}
  \end{equation}
  and
  \begin{equation}
    \begin{aligned}\label{eq2.7ww2}
          _{_R}\mathcal{D}_{h,1,0,-1}^{\alpha,\gamma_{1},\gamma_{2},\gamma_{3}}u(x)
      &:=\frac{\gamma_1}{h^\alpha}\sum_{k=0}^{[\frac{b-x}{h}]
      +1}w^{(\alpha)}_{k}e^{-(k-1)h\lambda}u(x+(k-1)h)\\
      &~~+\frac{\gamma_{2}}{h^\alpha}\sum_{k=0}^{[\frac{b-x}{h}]}w^{(\alpha)}_{k}e^{-kh\lambda}u(x+kh)
      +\frac{\gamma_3}{h^\alpha}\sum_{k=0}^{[\frac{b-x}{h}]
      -1}w^{(\alpha)}_{k}e^{-(k+1)h\lambda}u(x+(k+1)h)\\
      &~~-\frac{1}{h^\alpha}
      \big((\gamma_{1}e^{h\lambda}+\gamma_{2}+\gamma_{3}e^{-h\lambda})(1-e^{-h\lambda})^\alpha\big)u(x),
    \end{aligned}
  \end{equation}
where the parameters $\gamma_{j},j=1,2,3,$ satisfy the following
linear system
 \begin{equation}\label{lambdaparameter}
\begin{cases}
\displaystyle
\gamma_{1}+\gamma_{2}+\gamma_{3}=1,\\
\displaystyle \gamma_{1}-\gamma_{3}=\frac{\alpha}{2}.
\end{cases}
\end{equation}

The system (\ref{lambdaparameter}) has infinitely many solutions.
With the help of the knowledge of linear algebra, the solutions of
the system of linear algebraic equations (\ref{lambdaparameter}) can
be collected by the following three sets
 \begin{equation}\label{12parameters}
\mathcal{S}^{\alpha}_{1}(\gamma_{1},\gamma_{2},\gamma_{3})=\Big\{\gamma_1~
\textrm{is~given},\gamma_2=\frac{2+\alpha}{2}-2\gamma_{1},
~\gamma_3=\gamma_1-\frac{\alpha}{2}\Big\};
\end{equation}
 or
 \begin{equation}\label{02parameters}
\mathcal{S}^{\alpha}_{2}(\gamma_{1},\gamma_{2},\gamma_{3})=\Big\{\gamma_1=\frac{2+\alpha}{4}-\frac{\gamma_2}{2},~\gamma_2~
\textrm{is~given},
~\gamma_3=\frac{2-\alpha}{4}-\frac{\gamma_2}{2}\Big\};
\end{equation}
or
\begin{equation}\label{2parameters}
\mathcal{S}^{\alpha}_{3}(\gamma_{1},\gamma_{2},\gamma_{3})=\Big\{\gamma_1=\frac{\alpha}{2}+\gamma_3,
~\gamma_2=\frac{2-\alpha}{2}-2\gamma_3,~\gamma_3~
\textrm{is~given}\Big\}.
\end{equation}
The parameter values presented in the sets
$\mathcal{S}^{\alpha}_{j},j=1,2,3$ produce infinite number of second
order approximations for the Riemann-Liouville tempered fractional
derivative. Particularly, if taking $\lambda=0$ and $\gamma_j=0$
 in $\mathcal{S}^{\alpha}_{j},\,j=1,2,3$, they recover the second order
approximations presented in \cite{Tian:11} for the Riemann-Liouville
fractional derivative. After rearranging the weights
$w_k^{(\alpha)}$,  the Riemann-Liouville tempered fractional
derivatives at point $x_{j}$ are approximated as
\begin{equation}
  \begin{aligned}\label{threeGLD}
     _aD_x^{\alpha,\lambda}u(x_{j})-\alpha\lambda^{\alpha}u(x_j)=&
     \frac{1}{h^\alpha}\sum_{k=0}^{j+1}g_k^{(\alpha)}u(x_{j-k+1})
    -\frac{1}{h^\alpha}
      \big((\gamma_{1}e^{h\lambda}+\gamma_{2}+
      \gamma_{3}e^{-h\lambda})(1-e^{-h\lambda})^\alpha\big)u(x_j)\\
      &+O(h^2), \\
     _xD_b^{\alpha,\lambda}u(x_{j})-\alpha\lambda^{\alpha}u(x_j)=&
     \frac{1}{h^\alpha}\sum_{k=0}^{N-j+1}g_k^{(\alpha)}u(x_{j+k-1})
    -\frac{1}{h^\alpha}
      \big((\gamma_{1}e^{h\lambda}+\gamma_{2}
      +\gamma_{3}e^{-h\lambda})(1-e^{-h\lambda})^\alpha\big)u(x_j)\\
      &+O(h^2),
  \end{aligned}
\end{equation}
where the weights are given as
\begin{equation}\label{relation}
\begin{split}
\displaystyle
g_0^{(\alpha)}=\gamma_{1}w^{(\alpha)}_{0}e^{h\lambda},
~g_1^{(\alpha)}=\gamma_{1}w^{(\alpha)}_{1}+\gamma_{2}w_0^{(\alpha)},\\
\displaystyle
g_k^{(\alpha)}=\big(\gamma_{1}w^{(\alpha)}_{k}
+\gamma_{2}w^{(\alpha)}_{k-1}
    +\gamma_{3}w_{k-2}^{(\alpha)}\big)e^{-(k-1)h\lambda},~k\ge2.
\end{split}
\end{equation}
\begin{remark}\label{remtGL}
Similarly, for the Riemann-Liouville tempered
fractional derivatives defined in Definition \ref{def33}, we have
the second order difference approximations,
\begin{equation}
    \begin{aligned}\label{regleq2.7ww1}
       _{_L}\tilde{\mathcal{D}}_{h,1,0,-1}^{\alpha,\gamma_{1},\gamma_{2},\gamma_{3}}u(x)
      &=\frac{\gamma_1}{h^\alpha}\sum_{k=0}^{[\frac{x-a}{h}]+1}w^{(\alpha)}_{k}e^{-(k-1)h\lambda}u(x-(k-1)h)\\
      &~~+\frac{\gamma_2}{h^\alpha}\sum_{k=0}^{[\frac{x-a}{h}]}w^{(\alpha)}_{k}e^{-kh\lambda}u(x-kh)
      \\&~~+\frac{\gamma_3}{h^\alpha}\sum_{k=0}^{[\frac{x-a}{h}]-1}w^{(\alpha)}_{k}e^{-(k+1)h\lambda}u(x-(k+1)h);
 \end{aligned}
  \end{equation}
  and
  \begin{equation}
    \begin{aligned}\label{regleq2.7ww2}
          _{_R}\tilde{\mathcal{D}}_{h,1,0,-1}^{\alpha,\gamma_{1},\gamma_{2},\gamma_{3}}u(x)
      &=\frac{\gamma_1}{h^\alpha}\sum_{k=0}^{[\frac{b-x}{h}]
      +1}w^{(\alpha)}_{k}e^{-(k-1)h\lambda}u(x+(k-1)h)\\
      &~~+\frac{\gamma_{2}}{h^\alpha}\sum_{k=0}^{[\frac{b-x}{h}]}w^{(\alpha)}_{k}e^{-kh\lambda}u(x+kh)
      \\&~~+\frac{\gamma_3}{h^\alpha}\sum_{k=0}^{[\frac{b-x}{h}]
      -1}w^{(\alpha)}_{k}e^{-(k+1)h\lambda}u(x+(k+1)h).
    \end{aligned}
  \end{equation}
   After rearranging the weights
$w_k^{(\alpha)}$,  the Riemann-Liouville tempered fractional
derivatives at point $x_{j}$ are approximated as
\begin{equation}
  \begin{aligned}\label{reglthreeGLD}
     _aD_x^{\alpha,\lambda}u(x_{j})=&
     \frac{1}{h^\alpha}\sum_{k=0}^{j+1}g_k^{(\alpha)}u(x_{j-k+1})+O(h^2), \\
     _xD_b^{\alpha,\lambda}u(x_{j})=&
     \frac{1}{h^\alpha}\sum_{k=0}^{N-j+1}g_k^{(\alpha)}u(x_{j+k-1})+O(h^2),
  \end{aligned}
\end{equation}
where $g_k^{(\alpha)}$ is given in (\ref{relation}).
\end{remark}

\begin{figure}[h] 
    \subfigure[The bounds of $\gamma_1$ in set $\mathcal{S}^{\alpha}_{1}$.]{
        \label{fig:mini:subfig:a} 
        \begin{minipage}[b]{0.35\textwidth}
            \centering
            \includegraphics[scale=.38]{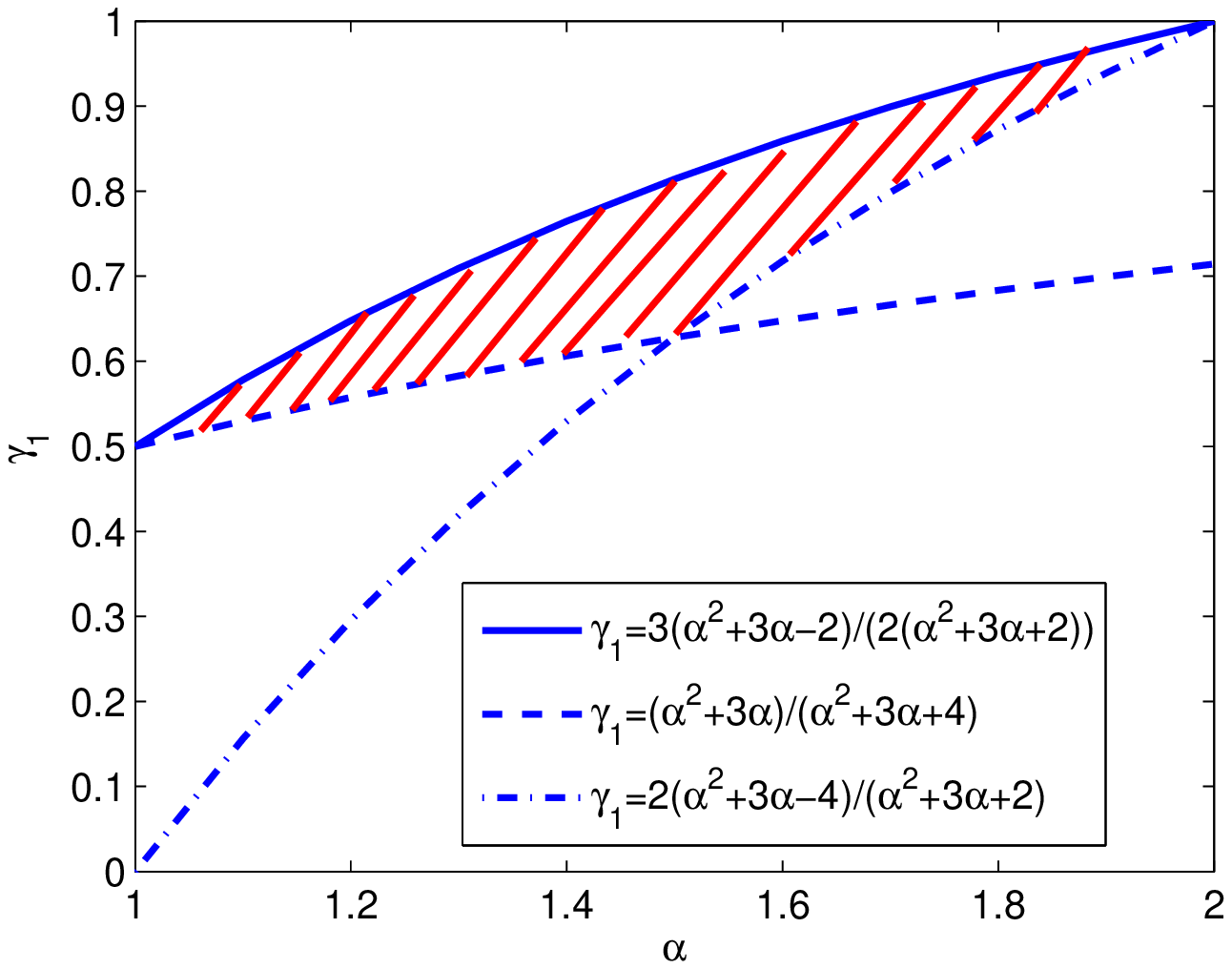}
        \end{minipage}}\hspace*{-20pt}%
    \subfigure[ The bounds of $\gamma_2$ in set $\mathcal{S}^{\alpha}_{2}$.]{
        \label{fig:mini:subfig:b} 
        \begin{minipage}[b]{0.35\textwidth}
            \centering
            \includegraphics[scale=.38]{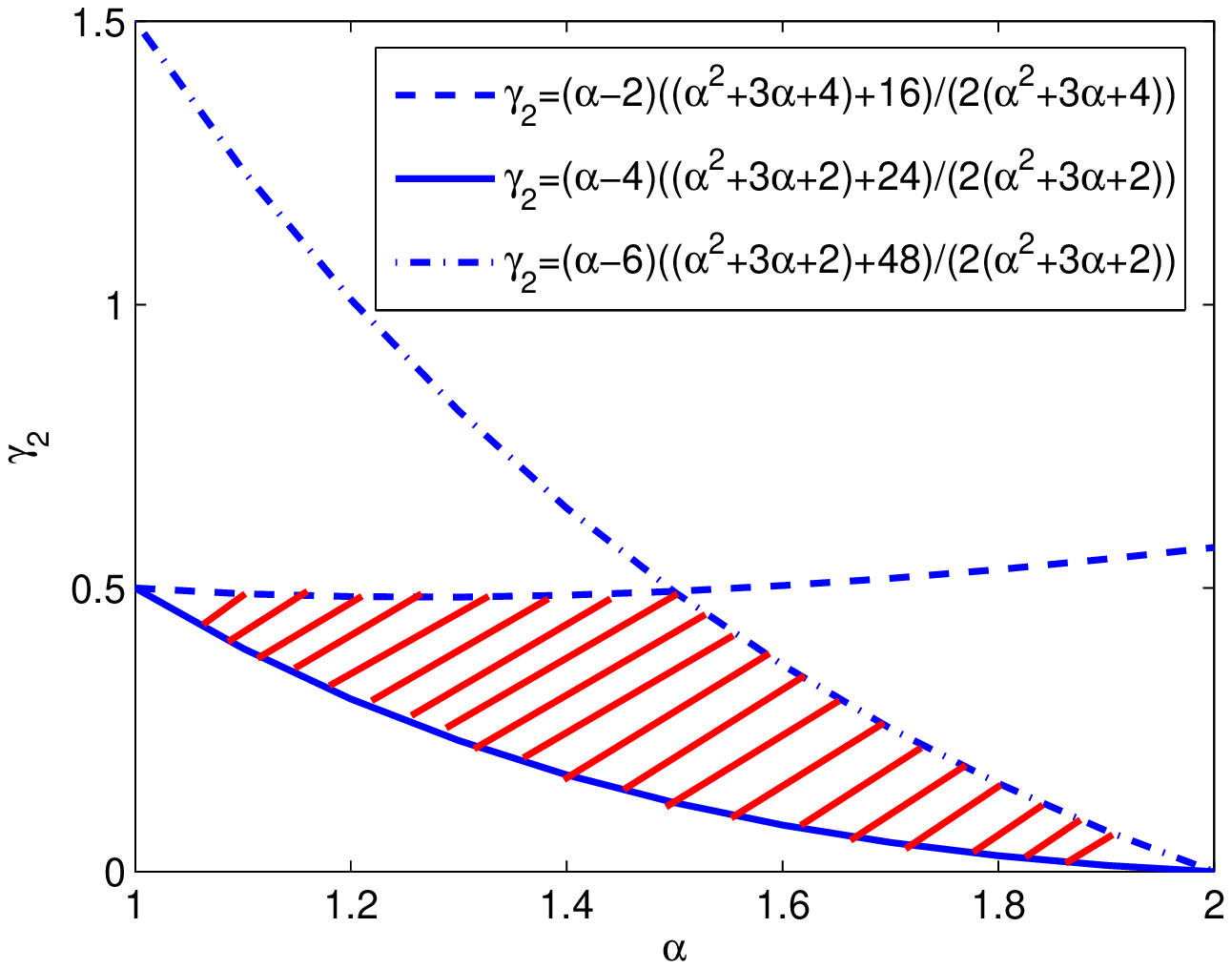}
        \end{minipage}}\hspace*{-20pt}%
\subfigure[ The bounds of $\gamma_3$ in set $\mathcal{S}^{\alpha}_{3}$.]{
        \label{fig:mini:subfig:b} 
        \begin{minipage}[b]{0.35\textwidth}
            \centering
            \includegraphics[scale=.38]{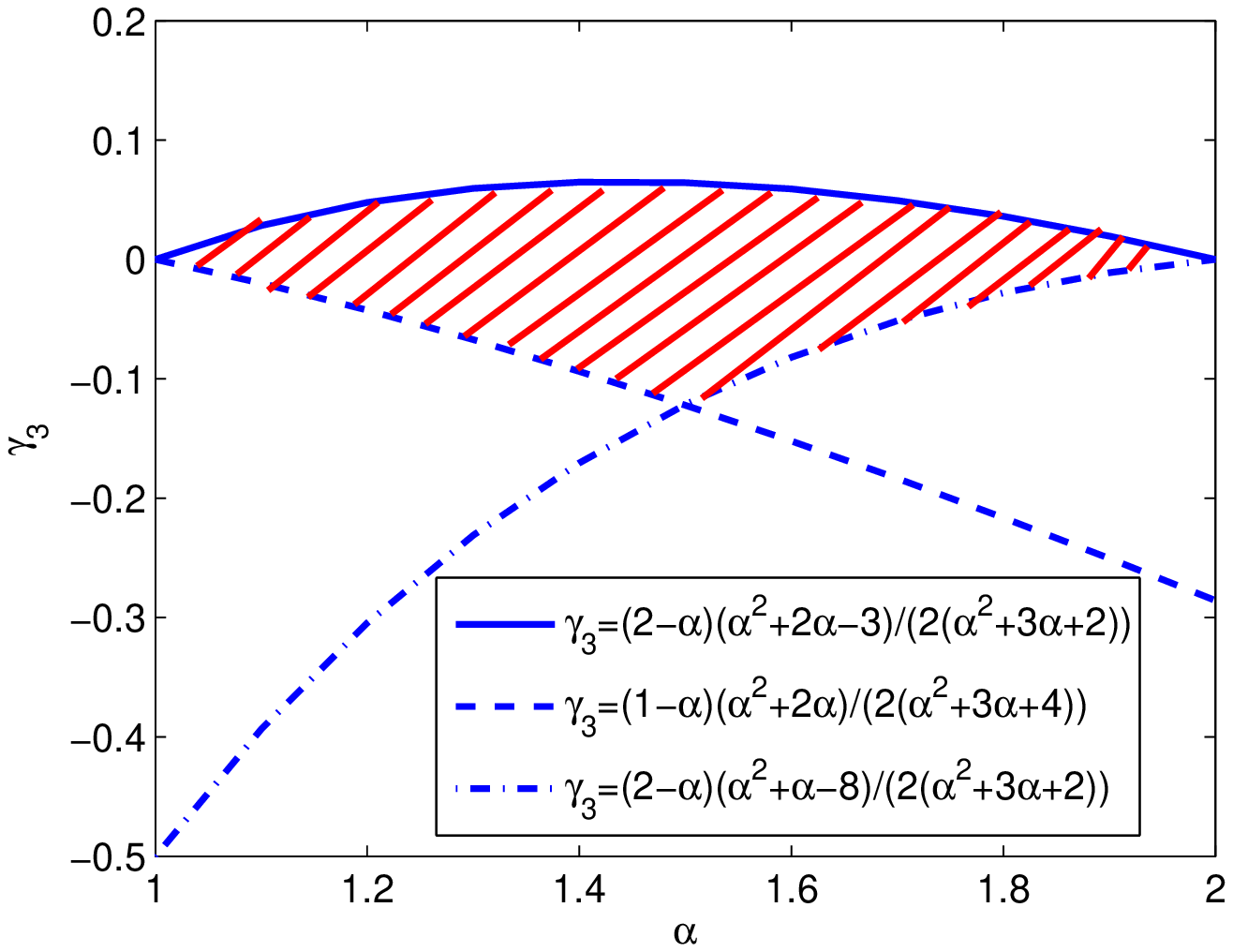}
        \end{minipage}}\hspace*{-20pt}%
    \caption{The bounds of $\gamma_1,\gamma_2,$ and $\gamma_3$  described in Lemma \ref{ww}.}\label{fig:one}
\end{figure}
\begin{lemma}\label{ww}
The weights appeared in (\ref{relation}) with $1<\alpha<2$ satisfy

(1).  $ w^{(\alpha)}_{0}=1,\, w^{(\alpha)}_{1}=-\alpha<0,\,w^{(\alpha)}_{k}=\big(1-\frac{1+\alpha}{k}\big)w^{(\alpha)}_{k-1}\,(k\geq1)$; 
$1\geq w^{(\alpha)}_{2}\geq w^{(\alpha)}_{3}\geq\ldots\geq0,\,\sum\limits_{k=0}^{\infty}w^{(\alpha)}_{k}= 0,\,\sum\limits_{k=0}^{m}w^{(\alpha)}_{k}< 0~(m\geq 1)$;

(2). For\, $h>0,\,\lambda\geq0$, if $\gamma_{1},\gamma_{2}$ and $\gamma_{3}$ are chosen in set
$\mathcal{S}^{\alpha}_{1}(\gamma_1,\gamma_2,\gamma_3)$
 with $\max\big\{\frac{2(\alpha^2+3\alpha-4)}{\alpha^2+3\alpha+2},
\frac{\alpha^2+3\alpha}{\alpha^2+3\alpha+4}\big\}\leq\gamma_1 \leq
\frac{3(\alpha^2+3\alpha-2)}{2(\alpha^2+3\alpha+2)}$,
 or\, set\,
$\mathcal{S}^{\alpha}_{2}(\gamma_1,\gamma_2,\gamma_3)$ with
$\frac{(\alpha-4)(\alpha^2+3\alpha+2)+24}{2(\alpha^2+3\alpha+2)}\leq\gamma_2
\leq
\min\big\{\frac{(\alpha-2)(\alpha^2+3\alpha+4)+16}{2(\alpha^2+3\alpha+4)},
\frac{(\alpha-6)(\alpha^2+3\alpha+2)+48}{2(\alpha^2+3\alpha+2)}\big\}$,
or set
$\mathcal{S}^{\alpha}_{3}(\gamma_1,\gamma_2,\gamma_3)$ with
$ \max\big\{\frac{(2-\alpha)(\alpha^2+\alpha-8)}{\alpha^2+3\alpha+2},
\frac{(1-\alpha)(\alpha^2+2\alpha)}{2(\alpha^2+3\alpha+4)}\big\}
\leq \gamma_3
\leq\frac{(2-\alpha)(\alpha^2+2\alpha-3)}{2(\alpha^2+3\alpha+2)}$,
then there exist
\begin{equation}\label{gg}
g_1^{(\alpha)}\leq0,\,g_2^{(\alpha)}+g_0^{(\alpha)}\geq0,\,g_k^{(\alpha)}\geq0\,\,(k\geq3).
\end{equation}
\end{lemma}
\begin{proof} For the proof of the first part of this lemma, one can see \cite{Meerschaert:04,Podlubny:99}.
For the second part of this lemma, we only prove the conclusion for $\gamma_{1},\gamma_{2}$ and $\gamma_{3}$ selected in set
$\mathcal{S}^{\alpha}_{1}(\gamma_1,\gamma_2,\gamma_3)$. The conclusions for $\gamma_{1},\gamma_{2}$ and $\gamma_{3}$ selected in sets $\mathcal{S}^{\alpha}_{2}(\gamma_1,\gamma_2,\gamma_3)$ and $\mathcal{S}^{\alpha}_{3}(\gamma_1,\gamma_2,\gamma_3)$ can be proved in a similar manner. According to \eqref{lambdaparameter},  we deduce that
\begin{equation}\label{ww1}
\begin{aligned}
g_1^{(\alpha)}&=\big(-\alpha\gamma_{1}
+\gamma_{2}\big)e^{-h\lambda}\\
&=\bigg(-(2+\alpha)\gamma_{1}
+\frac{2+\alpha}{2}\bigg)e^{-h\lambda}.
\end{aligned}
\end{equation}
Obviously, $\gamma_{1}\geq\frac{1}{2}$ implies $g_1^{(\alpha)}\leq0$.
Noting \eqref{lambdaparameter}, we see that
\begin{equation}\label{ww02}
\begin{aligned}
g_2^{(\alpha)}+g_0^{(\alpha)}&=\bigg(\frac{\alpha^2-\alpha}{2}\gamma_{1}
-\alpha\gamma_{2}+\gamma_{3}\bigg)e^{-h\lambda}+\gamma_{1}e^{h\lambda}\\
&\geq\bigg(\frac{\alpha^2-\alpha+2}{2}\gamma_{1}
-\alpha\gamma_{2}+\gamma_{3}\bigg)e^{-h\lambda}\\
&=\bigg(\frac{\alpha^2+3\alpha+4}{2}\gamma_{1}
-\frac{\alpha^2+3\alpha}{2}\bigg)e^{-h\lambda}\geq0,
 \end{aligned}
\end{equation}
 if $\gamma_{1}\geq\frac{\alpha^2+3\alpha}{\alpha^2+3\alpha+4}$.
In view of \eqref{lambdaparameter}, by a straightforward calculation, we obtain
\begin{equation}\label{ww3}
\begin{aligned}
g_3^{(\alpha)}&=\bigg(\frac{(2-\alpha)(\alpha-1)\alpha}{6}\gamma_{1}+\frac{\alpha^2-\alpha}{2}\gamma_{2}
-\alpha\gamma_{3}\bigg)e^{-2h\lambda}\\
&=\bigg(-\frac{\alpha(\alpha^2+3\alpha+2)}{6}\gamma_{1}
+\frac{\alpha(\alpha^2+3\alpha-2)}{4}\bigg)e^{-2h\lambda}\geq0,
 \end{aligned}
\end{equation}
 if $\gamma_{1}\leq\frac{3(\alpha^2+3\alpha-2)}{2(\alpha^2+3\alpha+2)}$.
More generally, for $k\geq4$, using the recurrence relation of $w^{(\alpha)}_{k}$, we have
\begin{equation}\label{wwk}
\begin{aligned}
g_k^{(\alpha)}&=\big(\gamma_{1}w^{(\alpha)}_{k}
+\gamma_{2}w^{(\alpha)}_{k-1}
    +\gamma_{3}w_{k-2}^{(\alpha)}\big)e^{-(k-1)h\lambda}\\
&=\bigg(\frac{(k-1-\alpha)(k-2-\alpha)}{k(k-1)}\gamma_{1}+\frac{k-2-\alpha}{k-1}\gamma_{2}
+\gamma_{3}\bigg)w_{k-2}^{(\alpha)}e^{-(k-1)h\lambda}\\
&=\bigg(\frac{\alpha^2+3\alpha+2}{k(k-1)}\gamma_{1}
+\frac{-\alpha^2-3\alpha+2k-4}{2(k-1)}\bigg)w_{k-2}^{(\alpha)}e^{-(k-1)h\lambda}\geq0,
 \end{aligned}
\end{equation}
if $\gamma_{1}\geq \frac{k(\alpha^2+3\alpha+4-2k)}{2(\alpha^2+3\alpha+2)}$.
It is easy to check that the bound $\frac{k(\alpha^2+3\alpha+4-2k)}{2(\alpha^2+3\alpha+2)}$
is decreasing with respect to the variable $k\,\,(k\geq 2)$ for $1<\alpha<2$.
 Combining the above formulas, we obtain the desired bounds of $\gamma_1$.
\end{proof}
\begin{remark}\label{remarkq}
The bounds of $\gamma_{1},\gamma_{2}$ and $\gamma_{3}$ are illustrated in Figure \ref{fig:one}.
For the Riemann-Liouville fractional calculus (i.e., $\lambda=0$), the restrictions for $\gamma_{1},\gamma_{2}$ and $\gamma_{3}$ given in Lemma \ref{ww} can be relaxed when using the generating function method \cite{Tian:11} to prove the numerical stability of time dependent PDE. And if the parameters $\gamma_{j},\,p_j,\,j=1,2,3,$ in (\ref{rWSGD}) and
(\ref{thm01}) are taken as
 \begin{equation}\label{Q12parametersa}
p_1=1,~p_2=0,~p_3=-1,~\gamma_1=\frac{3\alpha^2+5\alpha}{24},
~\gamma_2=\frac{12-3\alpha^2+\alpha}{12},~\gamma_{3}=\frac{3\alpha^2-7\alpha}{24},
\end{equation}
then the corresponding tempered-WSGD operators have third order
accuracy. It is easy to check that the parameters do not fall in the domains described in
Lemma \ref{ww}. Using the difference formulae \eqref{Q12parametersa} to approximate the
tempered fractional derivatives for fractional diffusion equations seems not to be stable.
In the next section, we select the stability ones to solve the time
dependent tempered fractional PDEs.
\end{remark}
\section{Numerical schemes for the tempered
fractional diffusion equation}\label{sec:3}
In this section, we apply the second order approximations of the
Riemann-Liouville tempered fractional derivative presented in \eqref{threeGLD} to the following tempered fractional diffusion equation
\begin{equation}\label{Problem}
  \begin{cases}
    \frac{\partial u(x,t)}{\partial t}=(l~{_{a}}\textbf{D}_x^{\alpha,\lambda}+r~{_x}\textbf{D}_{b}^{\alpha,\lambda})u(x,t)+s(x,t),
     &   \text{$(x,  t) \in (a,b)\times [0, T]$,} \\
    u(x, 0)=u_0(x),    & \text{$x\in (a,b)$},\\
    u(a,t)=\phi_{l}(t),~~u(b,t)=\phi_{r}(t), &\text{$t\in[0, T]$},
  \end{cases}
\end{equation}
where $u=u(x,t)$ is the concentration of a solute at a point $x$ at
time $t$; $s(x,t)$ is the source term;
and the weighting factors $l,r$ usually control the bias of the
dispersion. The diffusion coefficients $l$ and $r$ are nonnegative
constants with $l+r=1$. And if $l\neq0$, then $\phi_l(t)\equiv0$;
if $r\neq0$, then $\phi_r(t)\equiv0$. Next we discretize
 \eqref{Problem} by the second order tempered-WSGD operators given
 in \eqref{threeGLD}. In the following numerical analysis, we assume that \eqref{Problem} has an unique and
sufficiently smooth solution.

\subsection{Crank-Nicolson-tempered-WSGD schemes}
To derive the numerical schemes for problem \eqref{Problem}, we
first introduce some notations used later. The spatial interval $[a,
b]$ is divided into ${N_x}$ parts by the uniform mesh with the space
step $h=(b-a)/{N_x}$ and the temporal interval is partitioned into
${N_t}$ parts using the grid-points $t_n=n\tau$, where the
equidistant temporal step gives $\tau=T/{N_t}$. And the set of grid
points are denoted by $x_j=a+jh$ and $t_n=n\tau$ for $1\leq j\leq
{N_x}$ and $0\leq n\leq {N_t}$.
Denoting $t_{n+1/2}=(t_n+t_{n+1})/2$ and setting $
  u_j^n=u(x_j, t_n),u_j^n=(u_j^n+u_j^{n+1})/2, s_j^{n}=s(x_j, t_{n}),
$ we get the following Crank-Nicolson time discretization for
\eqref{Problem} at mesh point $(x_{j},t_{n})$:
\begin{equation*}
\frac{ u_j^{n+1}-u_j^n}{\tau}-\Big(l({_a}\textbf{D}_x^{\alpha,\lambda}u)_j^{n+1/2}
  +r({_x}\textbf{D}_b^{\alpha,\lambda}u)_j^{n+1/2} \Big)=s_j^{n+1/2}+O(\tau^2).
\end{equation*}
Using the tempered-WSGD operators
$_{_L}\mathcal{D}_{h,1,0,-1}^{\alpha,\gamma_{1},\gamma_{2},\gamma_{3}}u(x,
t)$ and
${_R}\mathcal{D}_{h,1,0,-1}^{\alpha,\gamma_{1},\gamma_{2},\gamma_{3}}u(x,t)$
to approximate the space Riemann-Liouville tempered fractional
derivatives ${_a}\textbf{D}_{x}^{\alpha,\lambda}u(x,t)$ and
${_x}\textbf{D}_{b}^{\alpha,\lambda}u(x,t)$, respectively, yields
\begin{equation}\label{scheme}
\begin{aligned}
     & \frac{ u_j^{n+1}-u_j^n}{\tau}-\Big(l~_{_L}\mathcal{D}_{h,1,0,-1}^{\alpha,\gamma_{1},\gamma_{2},\gamma_{3}}u_j^{n+1/2}
      +r~{_R}\mathcal{D}_{h,1,0,-1}^{\alpha,\gamma_{1},\gamma_{2},\gamma_{3}}u_j^{n+1/2}\Big)+\alpha \lambda^{\alpha-1}(l-r)~\delta_xu_{j}^{n+1/2}\\
      &=s_j^{n+1/2}+O(\tau^2+
      h^2),
\end{aligned}
\end{equation}
where $\delta_xu_{j}^{n}=(u_{j+1}^{n}-u_{j-1}^{n})/{2h}$.
Rearranging the above discretization (\ref{scheme}) leads to
\begin{equation}\label{Eq3.3}
\begin{aligned}
     &u_j^{n+1}-l~\tau
_{_L}\mathcal{D}_{h,1,0,-1}^{\alpha,\gamma_{1},\gamma_{2},\gamma_{3}}u_j^{n+1/2}
-r~\tau
{_R}\mathcal{D}_{h,1,0,-1}^{\alpha,\gamma_{1},\gamma_{2},\gamma_{3}}u_j^{n+1/2}
+\tau \alpha \lambda^{\alpha-1}(l-r)~\delta_xu_{j}^{n+1/2}\\
      &=u_j^n+
\tau s_j^{n+1/2}+O(\tau^3+\tau h^2).
\end{aligned}
\end{equation}
From \eqref{threeGLD}, we can recast (\ref{Eq3.3}) as
\begin{equation}\label{scheme1}
\begin{aligned}
     & u_j^{n+1}-\frac{l\tau}{h^{\alpha}}\sum_{k=0}^{j+1}g_k^{(\alpha)}u_{j-k+1}^{n+1/2}
        -\frac{r\tau}{h^{\alpha}}\sum_{k=0}^{{N_x}-j+1}g_k^{(\alpha)}u_{j+k-1}^{n+1/2}+\tau \alpha \lambda^{\alpha-1}(l-r)~\delta_xu_{j}^{n+1/2}\\
      &=u_j^n+\tau s_j^{n+1/2}+O(\tau^3+\tau h^2).
\end{aligned}
\end{equation}
Denoting $U_j^n$ as the numerical approximation of $u_j^n$ and
omitting the local truncation errors, we get the
Crank-Nicolson-tempered-WSGD scheme of \eqref{Problem} being given
by
\begin{equation}\label{scheme1I}
      U_j^{n+1}-\frac{l~\tau}{h^{\alpha}}\sum_{k=0}^{j+1}g_k^{(\alpha)}U_{j-k+1}^{n+1/2}
        -\frac{r~\tau}{h^{\alpha}}\sum_{k=0}^{{N_x}-j+1}g_k^{(\alpha)}U_{j+k-1}^{n+1/2} \\
      +\tau \alpha \lambda^{\alpha-1}(l-r)~\frac{U_{j+1}^{n+1/2}-U_{j-1}^{n+1/2}}{2h}=U_j^n+\tau s_j^{n+1/2}.
\end{equation}
For the convenience of implementation, we also introduce the matrix
form of the grid functions
\begin{equation*}
  U^n=\Big(U_1^n, U_2^n,\cdots, U_{{N_x}-1}^n\Big)^{\mathrm{T}}.
\end{equation*}
Then the numerical scheme \eqref{scheme1I} can be rewritten as
\begin{equation}\label{schemem}
  \Big(I-\frac{\tau}{2h^{\alpha}}(l~A+r~A^{\mathrm{T}})-\frac{\tau \alpha \lambda^{\alpha-1}(r-l)}{4h}B\Big)U^{n+1}
  =\Big(I+\frac{\tau}{2h^{\alpha}}(l~A+r~A^{\mathrm{T}})+\frac{\tau \alpha \lambda^{\alpha-1}(r-l)}{4h}B\Big)U^{n}+\tau F^{n+1/2},
\end{equation}
where the
matrix $A=\big(a_{m,j}\big)_{N_x-1,N_x-1}$ with the entries
  \begin{equation}\label{matrixA}
  a_{m,j}=
  \begin{cases}
    0,&   \text{$j>m+1$,} \\
    g_{0}^{(\alpha)},    & \text{$j=m+1$},\\
    g_{1}^{(\alpha)}
    -(l+r)(\gamma_{1}e^{h\lambda}+\gamma_{2}+\gamma_{3}e^{-h\lambda})(1-e^{-h\lambda})^\alpha,    & \text{$j=m$},\\
    g_{2}^{(\alpha)},    & \text{$j=m-1$},\\
    g_{m-j+1}^{(\alpha)}, &\text{$j\leq m-2$},
  \end{cases}
\end{equation}
 and  $B=tridiag\{-1,0,1\},$
is a symmetric tri-diagonal matrix of $N_x-1$-square.
The term $F^{n+1/2}$ gives
\begin{equation*}
\begin{aligned}
F^{n+1/2}&=\begin{pmatrix}
   s_{1}^{n+1/2}\\
     s_{2}^{n+1/2}\\
   \vdots \\
     s_{{N_x}-2}^{n+1/2}\\
     s_{{N_x}-1}^{n+1/2}
 \end{pmatrix}+
 \frac{U_{0}^{n+1/2}}{2h^{\alpha}}
 \begin{pmatrix}
   l~g_2^{(\alpha)}+r~g_0^{(\alpha)}\\
   l~g_3^{(\alpha)}\\
   \vdots \\
   l~g_{{N_x}-1}^{(\alpha)}\\
   l~g_{N_x}^{(\alpha)}
 \end{pmatrix}+
 \frac{U_{{N_x}}^{n+1/2}}{2h^{\alpha}}
 \begin{pmatrix}
   r~g_{N_x}^{(\alpha)}\\
   r~g_{{N_x}-1}^{(\alpha)}\\
   \vdots \\
   r~g_3^{(\alpha)}\\
   l~g_0^{(\alpha)}+r~g_2^{(\alpha)}
 \end{pmatrix}
+\frac{\alpha \lambda^{\alpha-1}(r-l)}{4h}
 \begin{pmatrix}
   U_{{0}}^{n+1/2}\\
   0\\
   \vdots \\
   0\\
   -U_{{N_x}}^{n+1/2}
 \end{pmatrix}.
\end{aligned}
\end{equation*}

\subsection{Stability and convergence}
Now we discuss the numerical stability and convergence for the
Crank-Nicolson-tempered-WSGD schemes \eqref{scheme1I}. We explore
the properties of the eigenvalues of the iterative matrix of
\eqref{scheme1I} on the grid points
$\{x_j=a+jh,\,h=(b-a)/N_{x},j=1,2,\ldots,N_{x}-1\}$. If the real parts
of the eigenvalues are negative, then the schemes are stable. First,
we introduce several lemmas.

\begin{lemma}[\cite{Quarteroni:07}]\label{lem3}
  A real matrix $A$ of order $n$ is positive definite if and
 only if its symmetric part $H=\frac{A+A^T}{2}$ is positive definite;
 $H$ is positive definite if and only if the eigenvalues of $H$ are positive.
\end{lemma}
\begin{lemma}[\cite{Quarteroni:07}]\label{thm:2}
  If $A\in\mathbb{C}^{n\times n}$, let $H=\frac{A+A^*}{2}$
  be the hermitian part of $A$, $A^*$ the conjugate transpose of $A$,
 then for any eigenvalue $\mu$ of $A$,
  there exists
  \begin{equation*}
    \mu_{\min}(H)\le\mathrm{Re}(\mu(A))\le\mu_{\max}(H),
  \end{equation*}
  where $\mathrm{Re(\mu(A))}$ represents the real part of $\mu$, and
  $\mu_{\min}(H)$ and
  $\mu_{\max}(H)$ are the minimum and maximum of the eigenvalues of $H$.
\end{lemma}
\begin{theorem}\label{thmm}
 Let the martrices $A=\big(a_{m,j}\big)_{N_x,N_x}, A^{T}=\big(a_{j,m}\big)_{N_x,N_x} $ be given in numerical scheme \eqref{schemem}.
If $\gamma_{1},\gamma_{2}$ and $\gamma_{3}$ are chosen in set
$\mathcal{S}^{\alpha}_{1}(\gamma_1,\gamma_2,\gamma_3)$
 with $\max\big\{\frac{2(\alpha^2+3\alpha-4)}{\alpha^2+3\alpha+2},
\frac{\alpha^2+3\alpha}{\alpha^2+3\alpha+4}\big\}\leq\gamma_1\leq
\frac{3(\alpha^2+3\alpha-2)}{2(\alpha^2+3\alpha+2)},$
 or\, set\,
$\mathcal{S}^{\alpha}_{2}(\gamma_1,\gamma_2,\gamma_3)$ with
$\frac{(\alpha-4)(\alpha^2+3\alpha+2)+24}{2(\alpha^2+3\alpha+2)}\leq\gamma_2
\leq
\min\big\{\frac{(\alpha-2)(\alpha^2+3\alpha+4)+16}{2(\alpha^2+3\alpha+4)},
\frac{(\alpha-6)(\alpha^2+3\alpha+2)+48}{2(\alpha^2+3\alpha+2)}\big\}$,
or set
$\mathcal{S}^{\alpha}_{3}(\gamma_1,\gamma_2,\gamma_3)$ with
$ \max\big\{\frac{(2-\alpha)(\alpha^2+\alpha-8)}{\alpha^2+3\alpha+2},
\frac{(1-\alpha)(\alpha^2+2\alpha)}{2(\alpha^2+3\alpha+4)}\big\}
\leq \gamma_3
\leq\frac{(2-\alpha)(\alpha^2+2\alpha-3)}{2(\alpha^2+3\alpha+2)}$,
 then the matrix $Q=\frac{A+A^T}{2}$ is diagonally dominant for $1<\alpha<2$ and all the eigenvalues of $Q$ are negative.
\end{theorem}
\begin{proof}
Denote $Q=\frac{A+A^T}{2}=\big(q_{m,j}\big)_{N_x-1,N_x-1}$ with the entries
  \begin{equation}\label{matrixA}
  q_{m,j}=
  \begin{cases}
    \frac{1}{2}g_{j-m+1}^{(\alpha)},&   \text{$j>m+1$,} \\
    \frac{1}{2}(g_{0}^{(\alpha)}+g_{2}^{(\alpha)}),    & \text{$j=m+1$},\\
    g_{1}^{(\alpha)}
    -(l+r)(\gamma_{1}e^{h\lambda}+\gamma_{2}+\gamma_{3}e^{-h\lambda})(1-e^{-h\lambda})^\alpha,    & \text{$j=m$},\\
    \frac{1}{2}(g_{2}^{(\alpha)}+g_{0}^{(\alpha)}),    & \text{$j=m-1$},\\
    \frac{1}{2}g_{m-j+1}^{(\alpha)}, &\text{$j\leq m-2$}.
  \end{cases}
\end{equation}
With the help of the following binomial formula
$$\sum^{\infty}_{m=0}w_{m}^{(\alpha)}e^{-mh\lambda}=(1-e^{-h\lambda})^\alpha,$$
we have
\begin{equation*}
    \begin{aligned}
            \sum^{\infty}_{m=0}g_{m}^{(\alpha)}&=
      \gamma_{1}w^{(\alpha)}_{0}e^{h\lambda}+\gamma_{1}w^{(\alpha)}_{1}+\gamma_{2}w_0^{(\alpha)}
+\sum^{\infty}_{m=2}\big(\gamma_{1}w^{(\alpha)}_{m}
+\gamma_{2}w^{(\alpha)}_{m-1}
    +\gamma_{3}w_{m-2}^{(\alpha)}\big)e^{-(m-1)h\lambda}
          \\&=(\gamma_{1}e^{h\lambda}+\gamma_{2}+\gamma_{3}e^{-h\lambda})(1-e^{-h\lambda})^\alpha.
    \end{aligned}
  \end{equation*}
Furthermore, we get
 $$\sum^{m+1}_{j=-\infty}q_{m,j}=-(\gamma_{1}e^{h\lambda}+\gamma_{2}+\gamma_{3}e^{-h\lambda})(1-e^{-h\lambda})^\alpha+
 \sum^{\infty}_{m=0}g_{m}^{(\alpha)}
 =0.$$
By a straightforward calculation, and using \eqref{lambdaparameter}, we get
$$\gamma_{1}e^{h\lambda}+\gamma_{2}+\gamma_{3}e^{-h\lambda}
=2\gamma_1\big(\cosh(h\lambda)-1\big)+1+\frac{\alpha}{2}(1-e^{-h\lambda})>0~ \text{ with}~1<\alpha<2~\text{and}~\gamma_1>0,$$
where $\cosh(h\lambda)$ denotes the hyperbolic cosine function $\coth(h\lambda)=\frac{e^{h\lambda}+e^{-h\lambda}}{2}$.
Noting that $\gamma_{1},\gamma_{2}$ and $\gamma_{3}$ are chosen in set
$\mathcal{S}^{\alpha}_{1}(\gamma_1,\gamma_2,\gamma_3)$
 or $\mathcal{S}^{\alpha}_{2}(\gamma_1,\gamma_2,\gamma_3)$
 or $\mathcal{S}^{\alpha}_{3}(\gamma_1,\gamma_2,\gamma_3)$, under the assumptions given in Lemma \ref{ww}, we obtain
 $q_{m,m}<0,m=1,2,\ldots,N_x$. Hence,
   $$-q_{m,m}>\sum^{m+1}_{j=0,j\neq m}q_{m,j},$$
  which implies that the matrix $Q$ is diagonally dominant.
Using the Gershgorin theorem \cite{Varga:00}, we deduce that the eigenvalues of matrix $Q$ are negative.
\end{proof}

\begin{theorem} \label{thmone}
Let $\gamma_{1},\gamma_{2}$ and $\gamma_{3}$ be chosen in set
$\mathcal{S}^{\alpha}_{1}(\gamma_1,\gamma_2,\gamma_3)$
 with $\max\big\{\frac{2(\alpha^2+3\alpha-4)}{\alpha^2+3\alpha+2},
\frac{\alpha^2+3\alpha}{\alpha^2+3\alpha+4}\big\}\leq\gamma_1\leq
\frac{3(\alpha^2+3\alpha-2)}{2(\alpha^2+3\alpha+2)},$
 or\, set\,
$\mathcal{S}^{\alpha}_{2}(\gamma_1,\gamma_2,\gamma_3)$ with
$\frac{(\alpha-4)(\alpha^2+3\alpha+2)+24}{2(\alpha^2+3\alpha+2)}\leq\gamma_2
\leq
\min\big\{\frac{(\alpha-2)(\alpha^2+3\alpha+4)+16}{2(\alpha^2+3\alpha+4)},
\frac{(\alpha-6)(\alpha^2+3\alpha+2)+48}{2(\alpha^2+3\alpha+2)}\big\}$,
or set
$\mathcal{S}^{\alpha}_{3}(\gamma_1,\gamma_2,\gamma_3)$ with
$ \max\big\{\frac{(2-\alpha)(\alpha^2+\alpha-8)}{\alpha^2+3\alpha+2},
\frac{(1-\alpha)(\alpha^2+2\alpha)}{2(\alpha^2+3\alpha+4)}\big\}
\leq \gamma_3
\leq\frac{(2-\alpha)(\alpha^2+2\alpha-3)}{2(\alpha^2+3\alpha+2)}$, then the
Crank-Nicolson-tempered-WSGD scheme (\ref{scheme1I}) with $\lambda\geq0$ and $1<\alpha<2$
 is stable.
\end{theorem}
\begin{proof}
  Denote $M=\frac{\tau}{2h^{\alpha}}(l~A+r~A^{\mathrm{T}})+\frac{\tau \alpha \lambda^{\alpha-1}(r-l)}{4h}B$. Then the matrix form \eqref{schemem} of the scheme (\ref{scheme1I}) can be rewritten as
  \begin{equation}
    (I-M)U^{n+1}=(I+M)U^{n}+\tau F^{n+1/2}.
  \end{equation}
  If denote $\mu(M)$ as an eigenvalue of matrix $M$, then $\frac{1+\mu(M)}{1-\mu(M)}$ is the eigenvalue of matrix $(I-M)^{-1}(I+M)$. Combining Lemma \ref{lem3}, Lemma \ref{thm:2} and Theorem \ref{thmm} shows that the eigenvalues of matrix $\frac{M+M^{\mathrm{T}}}{2}=\frac{\tau(l+r)}{4h^{\alpha}}(A+A^{\mathrm{T}})=\frac{\tau(l+r)}{4h^{\alpha}}Q$ are negative and $\mathrm{Re(\mu(M))}<0$, which implies that $\big|\frac{1+\mu(M)}{1-\mu(M)}\big|<1$. Therefore, the spectral radius of matrix $(I-M)^{-1}(I+M)$ is less than one; then the numerical scheme \eqref{scheme1I} is unconditionally stable.
\end{proof}
Define $ V_h=\{v: \,v=\{v_{m}\} \text{ is a grid function defined on
} \{x_m=mh\}_{i=1}^{N_{_{x}}-1}\text{~and~}v_{_{0}}=v_{_{N_{x}}}=0\}. $
And we define the corresponding discrete $L^2$-norm $
   \|v\|_{h}=\big(h\sum_{m=1}^{N_{x}-1}v_{m}^2\big)^{1/2}
$ for all $v=\{v_{m}\}\in V_h$.
\begin{theorem}\label{covthm}
Denote $u_j^n$ as the exact solution of problem \eqref{Problem}, and
$U_j^n$ the solution of the numerical scheme \eqref{scheme1I}. Let $\gamma_{1},\gamma_{2}$ and $\gamma_{3}$ be chosen in set
$\mathcal{S}^{\alpha}_{1}(\gamma_1,\gamma_2,\gamma_3)$
 with $\max\big\{\frac{2(\alpha^2+3\alpha-4)}{\alpha^2+3\alpha+2},
\frac{\alpha^2+3\alpha}{\alpha^2+3\alpha+4}\big\}\leq\gamma_1\leq
\frac{3(\alpha^2+3\alpha-2)}{2(\alpha^2+3\alpha+2)},$
 or\, set\,
$\mathcal{S}^{\alpha}_{2}(\gamma_1,\gamma_2,\gamma_3)$ with
$\frac{(\alpha-4)(\alpha^2+3\alpha+2)+24}{2(\alpha^2+3\alpha+2)}\leq\gamma_2
\leq
\min\big\{\frac{(\alpha-2)(\alpha^2+3\alpha+4)+16}{2(\alpha^2+3\alpha+4)},
\frac{(\alpha-6)(\alpha^2+3\alpha+2)+48}{2(\alpha^2+3\alpha+2)}\big\}$,
or set
$\mathcal{S}^{\alpha}_{3}(\gamma_1,\gamma_2,\gamma_3)$ with
$ \max\big\{\frac{(2-\alpha)(\alpha^2+\alpha-8)}{\alpha^2+3\alpha+2},
\frac{(1-\alpha)(\alpha^2+2\alpha)}{2(\alpha^2+3\alpha+4)}\big\}
\leq \gamma_3
\leq\frac{(2-\alpha)(\alpha^2+2\alpha-3)}{2(\alpha^2+3\alpha+2)}$.
Then we get
\begin{equation}
\|u^n-U^n\|_{h}\leq c(\tau^2+h^2),~1\leq n\leq N_{t},
\end{equation}
where $c$ denotes a positive constant and $\|\cdot\|_{h}$ the discrete $L^2$-norm; $u^n$ stands for $(u_1^n,u_2^n,\cdots, u_{N_x-1}^n)^T$.
\end{theorem}
\begin{proof}
Let $e_j^n=u_j^n-U_j^n$. Combining \eqref{scheme1} and \eqref{scheme1I} leads to
  \begin{equation}\label{p1}
    (I-M)E^{n+1}=(I+M)E^{n}+\tau
    \rho^n,
  \end{equation}
  where
  \begin{equation*}
    E^n=\Big(u_1^n-U_1^n,u_2^n-U_2^n,\cdots,u_{N_{x}-1}^n-U_{N_{x}-1}^n\Big)^{\mathrm{T}},\,
    \rho^n=\Big(\rho_1^n,\rho_2^n,\cdots,\rho_{N_{x}-1}^n\Big)^{\mathrm{T}},
  \end{equation*}
and $\rho_j^n=O(\tau^3+\tau h^2)$ is the local truncation error.
The above equation can be rewritten as
  \begin{equation*}
  E^{n+1}=(I-M)^{-1}(I+M)e^n+(I-M)^{-1}\rho^{n},
\end{equation*}
Taking the discrete $L^2$-norm on both sides of the above equation leads to
   \begin{equation*}
  ||E^{n}||_{h}\leq ||(I-M)^{-1}(I+M)E^{n-1}||_{h}+ ||(I-M)^{-1}\rho^{n}||_{h}.
\end{equation*}
  Noting that $|\rho_j^n|\leq \tilde{c}(\tau^2+h^2)$ and with the similar argument presented in Theorem \ref{thmone}, we can prove that
 \begin{equation*}
 ||(I-M)^{-1}(I+M)||_{h}\leq 1,~||(I-M)^{-1}||_{h}\leq 1.
\end{equation*}
Therefore, we further find that
\begin{equation*}
 ||E^{n}||_{h}\leq ||(I-M)^{-1}(I+M)E^{n-1}||_{h}+ ||\rho^{n}||_{h} \leq  ||E^{n-1}||_{h}+||\rho^{n}||_{h}.
\end{equation*}
Since the truncation error gives
  $|\rho_j^{n}|\leq C \tau(\tau^2+h^2)$ , we conclude that
\begin{equation*}
 ||E^{n}||_{h}\leq ||E^{n-1}||_{h}+||\rho^{n}||_{h}\leq\sum_{k=1}^{n}||\rho^{k}||_{h}\leq
 C (\tau^2+h^2).
\end{equation*}
\end{proof}
\section{Numerical results}\label{sec:4}
In this section, we perform the numerical experiments to verify the
approximation orders of the tempered-WSGD operators to the tempered
fractional calculus in Example 1; in Examples 2 and 3, to show the
powerfulness of the presented Crank-Nicolson-tempered-WSGD schemes
for the tempered fractional diffusion equations with the left
tempered fractional derivative and the right tempered fractional
derivative, respectively; in particular, the desired convergence
orders of the Crank-Nicolson-tempered-WSGD schemes are carefully
confirmed.

\begin{example}\label{example0}
We numerically test the approximation accuracy of the tempered-WSGD
operators to the left and right Riemann-Liouville tempered
fractional derivatives; and also the approximation
accuracy of the corresponding tempered-WSGD operators to the left
and right Riemann-Liouville tempered fractional integrals. Using
  \begin{equation*}\label{eq:5a}
  {_a}D_x^{\alpha}\big[(x-a)^{\mu}\big]=\frac{\Gamma(\mu+1)}{\Gamma(\mu-\alpha+1)}(x-a)^{\mu-\alpha},
 ~~ {_x}D_b^{\alpha}\big[(b-x)^{\mu}\big]=\frac{\Gamma(\mu+1)}{\Gamma(\mu-\alpha+1)}(b-x)^{\mu-\alpha},
  \end{equation*}
  \begin{equation*}\label{eq:5c}
  {_a}D_x^{-\sigma}\big[(x-a)^{\mu}\big]=\frac{\Gamma(\mu+1)}{\Gamma(\mu+\sigma+1)}(x-a)^{\mu+\sigma},~~
  {_x}D_b^{-\sigma}\big[(b-x)^{\mu}\big]=\frac{\Gamma(\mu+1)}{\Gamma(\mu+\sigma+1)}(b-x)^{\mu+\sigma},
  \end{equation*}
we obtain the analytical/exact results
  \begin{equation*}\label{eq:5ac}
  {_a}D_x^{\alpha,\lambda}\big[e^{-\lambda x}(x-a)^{\mu}\big]=\frac{\Gamma(\mu+1)}{\Gamma(\mu-\alpha+1)}e^{-\lambda x}(x-a)^{\mu-\alpha},
    ~~
  {_x}D_b^{\alpha,\lambda}\big[e^{\lambda x}(b-x)^{\mu}\big]=\frac{\Gamma(\mu+1)}{\Gamma(\mu-\alpha+1)}e^{\lambda x}(b-x)^{\mu-\alpha},
  \end{equation*}
  \begin{equation*}\label{eq:5cc}
  {_a}D_x^{-\sigma,\lambda}\big[e^{-\lambda x}(x-a)^{\mu}\big]=\frac{\Gamma(\mu+1)}{\Gamma(\mu+\sigma+1)}e^{-\lambda x}(x-a)^{\mu+\sigma},
    ~~
  {_x}D_b^{-\sigma,\lambda}\big[e^{\lambda x}(b-x)^{\mu}\big]=\frac{\Gamma(\mu+1)}{\Gamma(\mu+\sigma+1)}e^{\lambda
  x}(b-x)^{\mu+\sigma},
  \end{equation*}
  where $\mu>-1$.
\end{example}
The numerical values are computed in the finite interval $[0,1]$;
the numerical errors and orders of accuracy are shown in Tables
\ref{tabexm1a}-\ref{tabexm2a6}, which confirm the desired second
order accuracy.
\begin{table}[H]
  \begin{center}
  \caption{Numerical errors and orders of accuracy for $_0D_x^{\alpha,\lambda}(e^{-\lambda x}x^{2+\alpha})=\frac{\Gamma(3+\alpha)}{2}e^{-\lambda x}x^2$ computed by the tempered-WSGD operators (\ref{regleq2.7ww1})  for different $\lambda$  in the interval $[0,1]$ with fixed $\alpha=1.6$ and the $(\gamma_{1},\gamma_{2},\gamma_{3})$ are selected in set $\mathcal{S}^{\alpha}_{3}(\gamma_{1},\gamma_{2},\gamma_{3})$ with $\gamma_{3}=0.001$.}\vspace{5pt}
\begin{tabular}{|l|l|l|l|l|l|l|}
\hline
& \multicolumn{2}{c|}{$\lambda=0$} & \multicolumn{2}{c|}{$\lambda=1$} & \multicolumn{2}{c|}{$\lambda=10$}\\
\cline{2-3} \cline{4-5}\cline{6-7}
$h$ &  $\|e\|_{h}$-error  & order & $\|e\|_{h}$-error & order & $\|e\|_{h}$-error & order\\
\hline
   1/10   & 3.63e-03 &       &2.49e-03 &     &4.57e-04 & \\
   1/20   & 9.02e-04 & 2.01  &6.15e-04 & 2.02&1.22e-04 & 1.90 \\
   1/40   & 2.25e-04 & 2.00  &1.53e-04 & 2.01&3.04e-05 & 2.01 \\
   1/80   & 5.62e-05 & 2.00  &3.82e-05 & 2.00&7.54e-06 & 2.01 \\
\hline
\end{tabular}\label{tabexm1a}
\end{center}
\end{table}
\begin{table}[H]
  \begin{center}
  \caption{Numerical errors and orders of accuracy for $_xD_1^{\alpha,\lambda}(e^{\lambda x}(1-x)^{2+\alpha})=\frac{\Gamma(3+\alpha)}{2}e^{\lambda x}(1-x)^2$ computed by the tempered-WSGD operators (\ref{regleq2.7ww2}) for different $\lambda$ in the interval $[0,1]$ with fixed $\alpha=1.6$ and the $(\gamma_{1},\gamma_{2},\gamma_{3})$ are selected in set $\mathcal{S}^{\alpha}_{3}(\gamma_{1},\gamma_{2},\gamma_{3})$ with $\gamma_{3}=0$.}\vspace{5pt}
  \begin{tabular}{|l|l|l|l|l|l|l|}
\hline
& \multicolumn{2}{c|}{$\lambda=0$} & \multicolumn{2}{c|}{$\lambda=1$} & \multicolumn{2}{c|}{$\lambda=10$}\\
\cline{2-3} \cline{4-5}\cline{6-7}
$h$ &  $\|e\|_{h}$-error  & order & $\|e\|_{h}$-error & order & $\|e\|_{h}$-error & order\\
\hline
   1/10   & 3.63e-03 &       &6.75e-03 &     &1.01e+01 & \\
   1/20   & 9.02e-04 & 2.01  &1.67e-03 & 2.01&2.69e+00 & 1.90 \\
   1/40   & 2.25e-04 & 2.00  &4.16e-04 & 2.01&6.69e-01 & 2.01 \\
   1/80   & 5.62e-05 & 2.00  &1.04e-05 & 2.00&1.66e-01 & 2.01 \\
\hline
\end{tabular}\label{tabexm2a}
\end{center}
\end{table}
\begin{table}[H]
  \begin{center}\label{tabexm3a}
  \caption{Numerical errors and orders of accuracy for $_0D_x^{-\sigma,\lambda}(e^{-\lambda x}x^{1+\sigma})=\frac{\Gamma(2+\sigma)}{\Gamma(2+2\sigma)}e^{-\lambda x}x^{1+2\sigma}$  computed by the tempered-WSGD operators (replacing $\alpha$ by $-\sigma$ in  (\ref{regleq2.7ww1})) for different $\lambda$ in the interval $[0,1]$ with fixed $\sigma=0.6$ and the $(\gamma_{1},\gamma_{2},\gamma_{3})$ are selected in set $\mathcal{S}^{\sigma}_{3}(\gamma_{1},\gamma_{2},\gamma_{3})$ with $\gamma_{3}=0.04$.}\vspace{5pt}
  \begin{tabular}{|l|l|l|l|l|l|l|}
\hline
& \multicolumn{2}{c|}{$\lambda=0$} & \multicolumn{2}{c|}{$\lambda=2$} & \multicolumn{2}{c|}{$\lambda=5$}\\
\cline{2-3} \cline{4-5}\cline{6-7}
$h$ &  $\|e\|_{h}$-error  & order & $\|e\|_{h}$-error & order & $\|e\|_{h}$-error & order\\
\hline
   1/10   & 3.48e-03 &       &2.43e-03 &     &1.68e-03 & \\
   1/20   & 8.88e-04 & 1.97  &6.71e-04 & 1.86&5.30e-04 & 1.66 \\
   1/40   & 2.25e-04 & 1.98  &1.77e-04 & 1.92&1.49e-04 & 1.83 \\
   1/80   & 5.69e-05 & 1.98  &4.56e-05 & 1.96&3.98e-05 & 1.91 \\
\hline
\end{tabular}
\end{center}
\end{table}
\begin{table}[H]
  \begin{center}\label{tabexm4}
  \caption{Numerical errors and orders of accuracy for $_xD_1^{-\sigma,\lambda}(e^{\lambda x}(1-x)^{1+\sigma})=\frac{\Gamma(2+\sigma)}{\Gamma(2+2\sigma)}e^{\lambda x}(1-x)^{1+2\sigma}$ computed by the tempered-WSGD operators (replacing $\alpha$ by $-\sigma$ in  (\ref{regleq2.7ww2})) for different $\lambda$ in the interval $[0,1]$ with fixed $\sigma=0.6$ and the $(\gamma_{1},\gamma_{2},\gamma_{3})$ are selected in set $\mathcal{S}^{\sigma}_{3}(\gamma_{1},\gamma_{2},\gamma_{3})$ with $\gamma_{3}=-0.01$.}\vspace{5pt}
  \begin{tabular}{|l|l|l|l|l|l|l|}
\hline
& \multicolumn{2}{c|}{$\lambda=0$} & \multicolumn{2}{c|}{$\lambda=2$} & \multicolumn{2}{c|}{$\lambda=5$}\\
\cline{2-3} \cline{4-5}\cline{6-7}
$h$ &  $\|e\|_{h}$-error  & order & $\|e\|_{h}$-error & order & $\|e\|_{h}$-error & order\\
\hline
   1/10   & 4.35e-03 &       &2.22e-02 &     &3.05e-01 & \\
   1/20   & 1.11e-04 & 1.96  &6.16e-03 & 1.85&9.71e-02 & 1.66 \\
   1/40   & 2.84e-04 & 1.97  &1.63e-03 & 1.92&2.75e-02 & 1.82 \\
   1/80   & 7.18e-05 & 1.98  &4.22e-04 & 1.95&7.37e-03 & 1.90 \\
\hline
\end{tabular}
\end{center}
\end{table}
\begin{table}[H]
\centering
\caption{Numerical errors and orders of accuracy for $_0D_x^{\alpha,\lambda}(e^{-\lambda x}x^{2+\alpha})-\lambda^{\alpha}(e^{-\lambda x}x^{2+\alpha})$ computed by the tempered-WSGD operators (\ref{eq2.7ww1})  for different $\lambda$  in the interval $[0,1]$ and the $(\gamma_{1},\gamma_{2},\gamma_{3})$ are selected in set $\mathcal{S}^{\alpha}_{3}(\gamma_{1},\gamma_{2},\gamma_{3})$ with $\gamma_{3}=0.02$.}
\vspace{0.2cm}
  \begin{tabular}{|c|c|c|c|c|c|c|c|c|}
\hline
&  &\multicolumn{2}{c|}{$\lambda=0$} & \multicolumn{2}{c|}{$\lambda=1$} & \multicolumn{2}{c|}{$\lambda=10$}\\
\cline{3-4} \cline{4-5}\cline{6-8}
$\alpha$& $h$ &$\|e\|_{h}$-error  & order & $\|e\|_{h}$-error & order & $\|e\|_{h}$-error & order\\
\hline
   & 1/10  & 3.56e-03 &      &4.84e-03 &       &6.64e-04 & \\
   & 1/20  & 8.91e-04 & 2.00 &1.15e-03 & 2.08  &2.27e-04 & 1.55   \\
 $\alpha=0.5$
   & 1/40  & 2.23e-04 & 2.00 &2.85e-04 & 2.01  &6.15e-05 & 1.88\\
   & 1/80  & 5.57e-05 & 2.00 &7.11e-05 & 2.00  &1.47e-05 & 2.06 \\ \hline
   & 1/10  & 4.53e-03 &      &2.54e-03 &       &1.19e-04 & \\
 $\alpha=1.5$
   & 1/20  & 1.12e-03 & 2.02 &6.28e-04 & 2.02  &3.80e-05 & 1.65   \\
   & 1/40  & 2.79e-04 & 2.01 &1.56e-04 & 2.01  &1.05e-05 & 1.86\\
   & 1/80  & 6.96e-05 & 2.00 &3.90e-05 & 2.00  &2.72e-06 & 1.94 \\ \hline
  \end{tabular}
\label{tabexm1a5}
\end{table}
\begin{table}[H]
\centering
\caption{Numerical errors and orders of accuracy for $_xD_1^{\alpha,\lambda}(e^{\lambda x}(1-x)^{2+\alpha})-\lambda^\alpha(e^{\lambda x}(1-x)^{2+\alpha})$ computed by the tempered-WSGD operators (\ref{eq2.7ww2}) for different $\lambda$ in the interval $[0,1]$ and the $(\gamma_{1},\gamma_{2},\gamma_{3})$ are selected in set $\mathcal{S}^{\alpha}_{3}(\gamma_{1},\gamma_{2},\gamma_{3})$ with $\gamma_{3}=-0.02$.}
\vspace{0.2cm}
  \begin{tabular}{|c|c|c|c|c|c|c|c|c|}
\hline
&  &\multicolumn{2}{c|}{$\lambda=0$} & \multicolumn{2}{c|}{$\lambda=1$} & \multicolumn{2}{c|}{$\lambda=10$}\\
\cline{3-4} \cline{4-5}\cline{6-8}
$\alpha$& $h$ &$\|e\|_{h}$-error  & order & $\|e\|_{h}$-error & order & $\|e\|_{h}$-error & order\\
\hline
   & 1/10  & 2.45e-03 &      &8.62e-03 &       &1.30e+01 & \\
 $\alpha=0.5$
   & 1/20  & 6.19e-04 & 1.98 &2.15e-03 & 2.01  &3.97e+00 & 1.71   \\
   & 1/40  & 1.56e-04 & 1.99 &5.39e-04 & 1.99  &8.51e-01 & 2.22  \\
   & 1/80  & 3.91e-05 & 2.00 &1.35e-04 & 1.99  &2.23e-01 & 1.94 \\ \hline
   & 1/10  & 3.51e-03 &      &5.38e-03   &     &2.40e+00 & \\
 $\alpha=1.5$
   & 1/20  & 8.65e-04 & 2.02 &1.32e-03   & 2.03&7.14e-01 & 1.75   \\
   & 1/40  & 2.15e-04 & 2.01 &3.28e-04   & 2.01&1.88e-01 & 1.93\\
   & 1/80  & 5.38e-05 & 2.00 &8.19e-05   & 2.00&4.75e-02 & 1.98 \\ \hline
  \end{tabular}
\label{tabexm2a6}
\end{table}
\begin{example}\label{example1}
We consider the following tempered fractional diffusion equation
with the left tempered fractional derivative
 \begin{equation}\label{eq:55.1}
\begin{aligned}
       \frac{\partial u(x,t)}{\partial t}=&{_0}\mathbf{D}_x^{\alpha,\lambda}u(x,t)+e^{-\lambda x-t}\big((\lambda^{\alpha}-\alpha \lambda^{\alpha}-1)x^{1+\alpha}-\Gamma(2+\alpha)x\\
      &+\alpha(\alpha+1)\lambda^{\alpha-1}x^{\alpha}\big),\quad (x,t)\in (0,1)\times(0,1],\quad 1<\alpha<2,
\end{aligned}
\end{equation}
with  the boundary conditions
  \begin{equation*}
    u(0,t)=0,\ \ u(1,t)=e^{-\lambda-t}, \quad t\in [0, 1],
  \end{equation*}
  and the initial value
  \begin{equation*}
    u(x,0)=e^{-\lambda x}x^{1+\alpha}, \quad x\in [0, 1].
  \end{equation*}
We can check that the exact
solution of \eqref{eq:55.1} is $u(x,t)=e^{-\lambda
x-t}x^{1+\alpha}$.
\end{example}
Eq. \eqref{eq:55.1} is solved by the Crank-Nicolson-tempered-WSGD
scheme \eqref{scheme1I}; and the numerical results are collected in
Tables \ref{tab1}-\ref{tab3}. It can be seen
that the numerical results with second order accuracy are
obtained.
\begin{table}[H]
\caption{Numerical errors and orders of accuracy for Example
\ref{example1} computed by the Crank-Nicolson-tempered-WSGD
   schemes \eqref{scheme1I} at $t=1$ with different weights and the fixed stepsizes $\tau=h,\,\lambda=2.0 , \,\alpha=1.6$ and the parameters
   $(\gamma_{1},\gamma_{2},\gamma_{3})$ are selected in set $\mathcal{S}^{\alpha}_{1}(\gamma_{1},\gamma_{2},\gamma_{3})$. } \label{table.Ex2xi23}
\begin{center}
\begin{tabular}{|l|l|l|l|l|l|l|}
\hline
& \multicolumn{2}{c|}{$\gamma_{1}=0.7$} & \multicolumn{2}{c|}{$\gamma_{1}=0.75$} & \multicolumn{2}{c|}{$\gamma_{1}=0.8$}\\
\cline{2-3} \cline{4-5}\cline{6-7}
$h$ &  $\|e\|_{h}$-error  & order & $\|e\|_{h}$-error & order & $\|e\|_{h}$-error & order\\
\hline
   1/10   & 4.64e-04 &       &4.79e-04 &     &4.98e-04 & \\
   1/20   & 1.30e-04 & 1.84  &1.27e-04 & 1.92&1.25e-04 & 2.00 \\
   1/40   & 3.46e-05 & 1.91  &3.26e-05 & 1.96&3.08e-05 & 2.02 \\
   1/80   & 8.92e-06 & 1.95  &8.27e-06 & 1.98&7.63e-06 & 2.01 \\
\hline
\end{tabular}
\end{center}\label{tab1}
\end{table}
\begin{table}[H]
\caption{Numerical errors and orders of accuracy for Example
\ref{example1} computed by the Crank-Nicolson-tempered-WSGD
   schemes \eqref{scheme1I} at $t=1$ with different weights and the fixed stepsizes $\tau=h,\,\lambda=2.0, \,\alpha=1.6$ and the parameters
   $(\gamma_{1},\gamma_{2},\gamma_{3})$ are selected in set $\mathcal{S}^{\alpha}_{2}(\gamma_{1},\gamma_{2},\gamma_{3})$. } \label{table.Ex2xi23}
\begin{center}
\begin{tabular}{|l|l|l|l|l|l|l|}
\hline
& \multicolumn{2}{c|}{$\gamma_{2}=0.2$} & \multicolumn{2}{c|}{$\gamma_{2}=0.3$} & \multicolumn{2}{c|}{$\gamma_{2}=0.4$}\\
\cline{2-3} \cline{4-5}\cline{6-7}
$h$ &  $\|e\|_{h}$-error  & order & $\|e\|_{h}$-error & order & $\|e\|_{h}$-error & order\\
\hline
   1/10   & 4.98e-04 &       &4.79e-04 &      &4.64e-04 & \\
   1/20   & 1.25e-04 & 2.00  &1.27e-04 & 1.92 &1.30e-04 & 1.84 \\
   1/40   & 3.07e-05 & 2.02  &3.26e-05 & 1.96 &3.46e-05 & 1.91 \\
   1/80   & 7.62e-06 & 2.01  &8.27e-06 & 1.98 &8.92e-06 & 1.96 \\
\hline
\end{tabular}
\end{center}\label{tab2}
\end{table}
\begin{table}[H]
\caption{Numerical errors and orders of accuracy for Example
\ref{example1} computed by the Crank-Nicolson-tempered-WSGD
   schemes \eqref{scheme1I} at $t=1$ with different weights and the fixed stepsizes $\tau=h,\,\lambda=2, \,\alpha=1.6$ and the parameters
   $(\gamma_{1},\gamma_{2},\gamma_{3})$ are selected in set  $\mathcal{S}^{\alpha}_{3}(\gamma_{1},\gamma_{2},\gamma_{3})$. } \label{table.Ex2xi23}
\begin{center}
\begin{tabular}{|l|l|l|l|l|l|l|}
\hline
& \multicolumn{2}{c|}{$\gamma_{3}=-0.04$} & \multicolumn{2}{c|}{$\gamma_{3}=0$} & \multicolumn{2}{c|}{$\gamma_{3}=0.04$}\\
\cline{2-3} \cline{4-5}\cline{6-7}
$h$ &  $\|e\|_{h}$-error  & order & $\|e\|_{h}$-error & order & $\|e\|_{h}$-error & order\\
\hline
   1/10   & 4.82e-04 &       &4.98e-04 &     &5.16e-04 & \\
   1/20   & 1.26e-04 & 1.93  &1.25e-04 & 2.00&1.23e-04 & 2.07 \\
   1/40   & 3.22e-05 & 1.97  &3.08e-05 & 2.02&2.94e-05 & 2.07 \\
   1/80   & 8.14e-06 & 1.99  &7.63e-06 & 2.01&7.13e-06 & 2.04 \\
\hline
\end{tabular}
\end{center}\label{tab3}
\end{table}

\begin{example}\label{example2}
Finally, we consider the following tempered fractional diffusion
equation with the right tempered fractional derivative
  \begin{equation}\label{eq:5.1}
\begin{aligned}
       \frac{\partial u(x,t)}{\partial t}=&{_x}\mathbf{D}_b^{\alpha,\lambda}u(x,t)+e^{\lambda x-t}\big((\lambda^{\alpha}-\alpha\lambda^{\alpha}-1)(1-x)^{1+\alpha}-\Gamma(2+\alpha)(1-x)\\
      &+\alpha(\alpha+1)\lambda^{\alpha-1}(1-x)^{\alpha}\big), \quad (x,t)\in (0,1)\times(0,1],,\quad 1<\alpha<2,
\end{aligned}
\end{equation}
  with the boundary conditions
  \begin{equation*}
    u(0,t)=e^{\lambda x}(1-x)^{1+\alpha},\ \ u(1,t)=0, \quad t\in [0, 1],
  \end{equation*}
  and the initial value
  \begin{equation*}
    u(x,0)=e^{\lambda x}(1-x)^{1+\alpha}, \quad x\in [0, 1].
  \end{equation*}
With the help of the formulae given in Example \ref{example0},
 we get the exact solution of \eqref{eq:5.1}: $u(x,t)=e^{\lambda x-t}(1-x)^{1+\alpha}$.
\end{example}
Tables \ref{tab4}-\ref{tab6} present the numerical errors and the
convergence behaviors of the Crank-Nicolson-tempered-WSGD schemes
\eqref{scheme1I}. These confirm the results given in Theorem
\ref{covthm}.
\begin{table}[H]
\caption{Numerical errors and orders of accuracy for Example
\ref{example2} computed by the Crank-Nicolson-tempered-WSGD
   schemes \eqref{scheme1I} at $t=1$ with different weights and the fixed stepsizes $\tau=h,\,\lambda=1.0, \,\alpha=1.2$ and the parameters
   $(\gamma_{1},\gamma_{2},\gamma_{3})$ are selected in set $\mathcal{S}^{\alpha}_{1}(\gamma_{1},\gamma_{2},\gamma_{3})$. } \label{table.Ex2xi23}
\begin{center}
\begin{tabular}{|l|l|l|l|l|l|l|}
\hline
& \multicolumn{2}{c|}{$\gamma_{1}=0.7$} & \multicolumn{2}{c|}{$\gamma_{1}=0.75$} & \multicolumn{2}{c|}{$\gamma_{1}=0.8$}\\
\cline{2-3} \cline{4-5}\cline{6-7}
$h$ &  $\|e\|_{h}$-error  & order & $\|e\|_{h}$-error & order & $\|e\|_{h}$-error & order\\
\hline
   1/10   & 3.94e-03 &       &4.18e-03 &     &4.43e-03 & \\
   1/20   & 9.22e-04 & 2.09  &9.53e-04 & 2.14&9.85e-04 & 2.17\\
   1/40   & 2.18e-04 & 2.07  &2.20e-04 & 2.12&2.22e-04 & 2.16 \\
   1/80   & 5.30e-05 & 2.04  &5.25e-05 & 2.06&5.21e-05 & 2.10 \\
\hline
\end{tabular}
\end{center}\label{tab4}
\end{table}
\begin{table}[H]
\caption{Numerical errors and orders of accuracy for Example
\ref{example2} computed by the Crank-Nicolson-tempered-WSGD
   schemes \eqref{scheme1I} at $t=1$ with different weights and the fixed stepsizes $\tau=h,\,\lambda=1.0, \,\alpha=1.2$ and the parameters
   $(\gamma_{1},\gamma_{2},\gamma_{3})$ are selected in set $\mathcal{S}^{\alpha}_{2}(\gamma_{1},\gamma_{2},\gamma_{3})$. } \label{table.Ex2xi23}
\begin{center}
\begin{tabular}{|l|l|l|l|l|l|l|}
\hline
& \multicolumn{2}{c|}{$\gamma_{2}=0.2$} & \multicolumn{2}{c|}{$\gamma_{2}=0.3$} & \multicolumn{2}{c|}{$\gamma_{2}=0.4$}\\
\cline{2-3} \cline{4-5}\cline{6-7}
$h$ &  $\|e\|_{h}$-error  & order & $\|e\|_{h}$-error & order & $\|e\|_{h}$-error & order\\
\hline
   1/10   & 3.94e-03 &       &3.69e-03 &     &3.46e-03 & \\
   1/20   & 9.22e-04 & 2.09  &8.95e-04 & 2.05&8.70e-04 & 1.99 \\
   1/40   & 2.18e-04 & 2.07  &2.17e-04 & 2.04&2.17e-04 & 2.01 \\
   1/80   & 5.29e-05 & 2.04  &5.35e-05 & 2.02&5.40e-05 & 2.00 \\
\hline
\end{tabular}
\end{center}\label{tab5}
\end{table}
\begin{table}[H]
\caption{Numerical errors and orders of accuracy for Example
\ref{example2} computed by the Crank-Nicolson-tempered-WSGD
   schemes \eqref{scheme1I} at $t=1$ with different weights and the fixed stepsizes $\tau=h,\,\lambda=1.0, \,\alpha=1.2$ and the parameters
   $(\gamma_{1},\gamma_{2},\gamma_{3})$ are selected in set $\mathcal{S}^{\alpha}_{3}(\gamma_{1},\gamma_{2},\gamma_{3})$. } \label{table.Ex2xi23}
\begin{center}
\begin{tabular}{|l|l|l|l|l|l|l|}
\hline
& \multicolumn{2}{c|}{$\gamma_{3}=-0.04$} & \multicolumn{2}{c|}{$\gamma_{3}=0$} & \multicolumn{2}{c|}{$\gamma_{3}=0.04$}\\
\cline{2-3} \cline{4-5}\cline{6-7}
$h$ &  $\|e\|_{h}$-error  & order & $\|e\|_{h}$-error & order & $\|e\|_{h}$-error & order\\
\hline
   1/10   & 3.29e-03 &       &3.46e-03 &     &3.65e-03 & \\
   1/20   & 8.53e-04 & 1.95  &8.70e-04 & 1.99&8.89e-04 & 2.04 \\
   1/40   & 2.16e-04 & 1.98  &2.17e-04 & 2.00&2.17e-04& 2.03 \\
   1/80   & 5.45e-05 & 1.99  &5.40e-05 & 2.00&5.36e-05 & 2.02 \\
\hline
\end{tabular}
\end{center}\label{tab6}
\end{table}
\section{Concluding remarks}\label{sec:6}
L\'{e}vy flight models suppose that the particles have very large
jumps; and they have infinite moments. But many realistically
non-Brownian (at least converge to the Brownian ultraslowly and it
is not possible to observe the Brownian behaviors in the finite
observing time) physical processes just lie in the bounded physical
domain. So some techniques to modify the L\'{e}vy flight models are
introduced. The most popular one seems to be exponentially tempering
the probability of large jumps of L\'{e}vy flight, which leads to
the tempered fractional diffusion equation being used to describe
the probability density function of the positions of the particles.
With this model, the tempered fractional calculus are introduced;
they are very similar to but still different from the fractional
substantial calculus. The fractional substantial calculus are
time-space coupled operators; and their discretizations are in the
time direction. The tempered fractional derivative used in this
paper is a space operator without coupling with time. On one hand,
we need to derive its high order discretizations, which can greatly
improve the accuracy but without introducing new computational cost
comparing with the first order scheme; on the other hand, the
numerical stability of the derived schemes is a key issue.

This paper derive a series of high order discetizations for the
tempered fractional calculus, including the left Riemann-Liouville
tempered fractional derivative and integral and the right
Riemann-Liouville tempered fractional derivative and integral. In
particular, the superconvergent point still exists for the first
order discretization of left/right Riemann-Liouville tempered
fractional derivative/integral. The stability domains of the schemes are analytically derived and clearly illustrated in figures. A family of
second order schemes are used to numerically solve the tempered
fractional diffusion equation. And the stability and convergence of
the numerical schemes are theoretically proved and numerically
verified.
\section*{Acknowledgements}
We would like to thank the anonymous referees for their careful reading of this paper and their many valuable comments and  suggestions for improving the presentation of this work. This research was partially supported by the
National Natural Science Foundation of China under Grant  No. 11271173, the Starting Research Fund from the Xi'an University of Technology under Grant No. 108-211206 and the Scientific Research Program Funded by Shaanxi Provincial Education Department under Grant No. 2013JK0581.


  \end{document}